%% file: pyFFS_pub.tex
\documentclass[onefignum,onetabnum]{siamart190516}

\input{ex_shared}

\input{definitions}

\ifpdf
\hypersetup{
  pdftitle={pyFFS: A Python Library for Fast Fourier Series Computation and Interpolation with GPU Acceleration},
  pdfauthor={E. Bezzam, S. Kashani, P. Hurley, and M. Vetterli and M. Simeoni}
}
\fi

\begin{document}

\maketitle

\input{content/abstract}


\input{content/intro}
\input{content/2_theory}
\input{content/3_pyffs}
\input{content/4_benchmark}

\input{content/5_fourier_optics}
\input{content/6_conclusion}

\input{content/appendix}

\bibliographystyle{siamplain}
\bibliography{references}
\end{document}

%% file: ex_shared.tex
\input{setup}

\newsiamremark{remark}{Remark}
\newsiamremark{hypothesis}{Hypothesis}
\crefname{hypothesis}{Hypothesis}{Hypotheses}
\newsiamthm{claim}{Claim}

\headers{\pyffs: Fast Fourier Series Computation with Python}{E. Bezzam, S. Kashani, P. Hurley, M. Vetterli, and M. Simeoni}

\title{\pyffs: A Python Library for Fast Fourier Series Computation and Interpolation with GPU Acceleration\thanks{Published on August 18, 2022.
\funding{This work was in part funded by the Swiss National Science Foundation (SNSF) under grants 200021\textunderscore181978/1 SESAM –- Sensing and Sampling: Theory and Algorithms (first, second and fifth authors) and CRSII5 180232 FemtoLippmann -- Digital twin for multispectral imaging (second and fifth authors). For part of this work, the second, third, and fifth authors were also with the Foundations of Cognitive Solutions group in the IBM Research Laboratory of Zurich.}
}
}

\author{Eric Bezzam\thanks{Audiovisual Communications Laboratory, \'{E}cole Polytechnique F\'{e}d\'{e}rale de Lausanne, Switzerland 
  (\email{eric.bezzam@epfl.ch}, \email{sepand.kashani@epfl.ch}, \email{martin.vetterli@epfl.ch}).}
\and Sepand Kashani\footnotemark[2]
\and Paul Hurley\thanks{Western Sydney University, Australia 
  (\email{p.hurley@westernsydney.edu.au}).}
\and Martin Vetterli\footnotemark[2]
\and Matthieu Simeoni\thanks{Centre for Imaging, \'{E}cole Polytechnique F\'{e}d\'{e}rale de Lausanne, Switzerland 
	(\email{matthieu.simeoni@epfl.ch}).}
}

\usepackage{amsopn}


%% file: setup.tex
\usepackage{lipsum}
\usepackage{graphicx}
\usepackage{epstopdf}
\usepackage{algorithmic}
\usepackage{subcaption}
\usepackage{amsmath,amssymb,amsfonts,mathtools,stmaryrd}
\ifpdf
  \DeclareGraphicsExtensions{.eps,.pdf,.png,.jpg}
\else
  \DeclareGraphicsExtensions{.eps}
\fi

\usepackage{listings}
\usepackage{xcolor}
\usepackage{url}
\usepackage{hyperref}

\definecolor{codegreen}{rgb}{0,0.6,0}
\definecolor{codegray}{rgb}{0.5,0.5,0.5}
\definecolor{codepurple}{rgb}{0.58,0,0.82}
\definecolor{backcolour}{rgb}{0.95,0.95,0.92}

\newcommand\rurl[1]{%
	\href{http://#1}{\nolinkurl{#1}}%
}

\lstdefinestyle{mystyle}{
	backgroundcolor=\color{backcolour},   
	commentstyle=\color{codegreen},
	keywordstyle=\color{magenta},
	numberstyle=\tiny\color{codegray},
	stringstyle=\color{codepurple},
	basicstyle=\ttfamily\footnotesize,
	breakatwhitespace=false,         
	breaklines=true,                 
	captionpos=b,                    
	keepspaces=true,                 
	numbersep=5pt,                  
	showspaces=false,                
	showstringspaces=false,
	showtabs=false,                  
	tabsize=4
}

\usepackage{array}
\usepackage[export]{adjustbox}
\newcolumntype{x}[1]{>{\centering\arraybackslash\hspace{0pt}}p{#1}}

\usepackage{pifont}

\crefname{lstlisting}{listing}{listings}
\Crefname{lstlisting}{Listing}{Listings}

\newcommand{\pyffs}{{\MakeLowercase{py}\MakeUppercase{FFS}}}

%% file: definitions.tex
\newcommand{\bbE}{{\bf{E}}}

\newcommand{\bbx}{{\bf{x}}}
\newcommand{\bbX}{{\bf{X}}}

\newcommand{\bbZero}{{\bf{0}}}

\newcommand{\bC}{\mathbb{C}}

\newcommand{\bN}{\mathbb{N}}

\newcommand{\bR}{\mathbb{R}}

\newcommand{\bZ}{\mathbb{Z}}

\newcommand{\cO}{\mathcal{O}}

\newcommand{\bigBrack}[1]{\left[ #1 \right]}
\newcommand{\bigCurly}[1]{\left\{ #1 \right\}}
\newcommand{\bigParen}[1]{\left( #1 \right)}

\DeclareMathOperator{\DFT}{DFT}
\DeclareMathOperator{\iDFT}{IDFT}
\DeclareMathOperator{\CZT}{CZT}

\DeclarePairedDelimiter{\ceil}{\lceil}{\rceil}

%% file: content/abstract.tex
\begin{abstract}
Fourier transforms are an often necessary component in many computational tasks, and can be computed efficiently through the fast Fourier transform (FFT) algorithm. 
However, many applications involve an underlying continuous signal, and a more natural choice would be to work with e.g.\ the Fourier series (FS) coefficients in order to avoid the additional overhead of translating between the analog and discrete domains. Unfortunately, there exists very little literature and tools for the manipulation of FS coefficients from discrete samples.
This paper introduces a Python library called pyFFS for efficient FS coefficient computation, convolution, and interpolation. 
While the libraries SciPy and NumPy provide efficient routines for discrete Fourier transform coefficients via the FFT algorithm, pyFFS addresses the computation of FS coefficients through what we call the fast Fourier series (FFS). 
Moreover, pyFFS includes an FS interpolation method based on the chirp Z-transform that can make it more than an order of magnitude faster than the SciPy equivalent when one wishes to perform distortionless bandlimited interpolation. GPU support through CuPy is readily available, and allows for further acceleration: an order of 
magnitude faster for computing the 2-D FS coefficients of $1000\times1000$ 
samples and nearly two orders of magnitude faster for 2-D interpolation.
As an application, we discuss the use of pyFFS in Fourier optics.
pyFFS is available as an open source package at https://github.com/imagingofthings/pyFFS, with documentation at https://pyffs.readthedocs.io.
%
\end{abstract}

%% file: content/intro.tex
\section{Introduction}

Discretization is an inevitable part of digital signal processing. Although the universe around us moves and shakes with infinite precision, our computers can only handle so much. However, with a useful model and the appropriate analog and digital processing, we can faithfully simulate and process continuous-domain processes.
This idea is the essence of the Nyquist-Shannon sampling theorem for bandlimited signals~\cite{Shannon1949} and more generally for signals that have a finite rate of innovation~\cite{Blu2008}. Simply put: a function $f(t)$ with finite degrees-of-freedom can be completely determined by a finite number of its samples. In his paper~\cite{Shannon1949}, Shannon even admits that the theorem so often associated with him is a fact ``which is common knowledge in the communication art.'' However, without the proper formulation, common knowledge may never seep into common practice. 

The inception of the fast Fourier transform (FFT) algorithm~\cite{Cooley1965} has a  similar story.
Arguably one of the most influential algorithms in computational and natural sciences, the FFT was nearly not published, as one of the original authors, John Tukey, felt that it was a simple observation that was probably already known and unworthy of attention.\footnote{As history would later tell, Tukey was indeed right as Gauss had written a paper about an interpolation technique with essentially the same idea~\cite{Gauss1866}.} With the digital revolution that ensued afterwards, the FFT quickly became an indispensable algorithm in digital signal processing. Once again, we see that without the proper formulation and implementation, common knowledge may not establish itself in common practice.

Continuing along these lines,
we formulate old Fourier analysis tricks  with a new perspective and propose new tools to efficiently compute and interpolate Fourier series (FS) coefficients. As is well known, for discrete samples of a periodic, bandlimited signal, we can perfectly recover the FS coefficients of the underlying function when sampling is done according to the Nyquist-Shannon sampling theorem \cite{vetterli2014foundations}. Moreover, continuous-domain operations, such as convolution and interpolation, can be implemented in terms of distortionless discrete operations on the FS coefficients. 
These mathematical facts can be used in practice to design an algorithm for fast FS computation and interpolation, which we call the \emph{fast Fourier series} (FFS) algorithm. While deceptively simple, the FFS algorithm is surprisingly neither described in signal processing textbooks nor implemented in numerical computing libraries. For example, NumPy~\cite{Harris2020} and SciPy~\cite{Virtanen2020}  in the Python ecosystem focus mainly on the discrete Fourier transform (DFT) and related operations. We aim to change this by providing an efficient and easy-to-use interface to the FFS algorithm via a Python package called pyFFS. Furthermore, distortionless discrete operations with pyFFS are not limited to periodic signals. In fact, one is often interested in working with signals of compact-support, which can be cast as periodic by repeating the signal with a period larger than or equal to its support.


So why use pyFFS rather than the FFT routines from NumPy/SciPy? One reason is convenience when working with continuous-domain signals. The philosophy of pyFFS is to retain the continuous-domain perspective, often neglected when using numerical libraries such as NumPy and SciPy, which allows for much clearer code as seen in \Cref{lst:prop}. This can prevent common pitfalls due to an invalid conversion between discrete and continuous domains. Moreover, FS  coefficients are an important component of the Non Uniform Fast Fourier Transform (NUFFT), extensively used in MRI imaging and other computational imaging modalities~\cite{barnett2019parallel}. Another reason is efficiency. We benchmark pyFFS with equivalent functions in SciPy, observing scenarios in which the proposed library is more than an order of magnitude faster, e.g.\ for interpolation. Moreover, GPU support has been seamlessly incorporated for an even faster implementation. Just as the FFT implementation via NumPy and SciPy can be readily used 
for an efficient $ \mathcal{O}(N \log N) $ analysis and synthesis of discrete sequences, pyFFS offers the same ease-of-use and performance capabilities for discrete representations of continuous-domain signals.

The paper is organized as follows. In \Cref{sec:theory}, we go over the theory behind pyFFS and motivate why working with FS coefficients may be preferable. We also present theorems showing how FS coefficients of periodic, bandlimited signals can be computed and interpolated exactly and efficiently. In \Cref{sec:pyffs}, the Python user interface of pyFFS is discussed.\footnote{A more extensive documentation can be found at \url{https://pyffs.readthedocs.io}.} \Cref{sec:benchmark} presents benchmarking tests showing the efficiency of pyFFS against SciPy and the gains of GPU acceleration. 
In \Cref{sec:results}, we discuss an application of pyFFS in Fourier optics. Another application of pyFFS can be found in \cite{fageot2020tv}, where the authors make use of the FFS algorithm to efficiently compute multidimensional periodic splines. \Cref{sec:conclusion} concludes the paper.

%% file: content/2_theory.tex
\section{Theory}
\label{sec:theory}

In this section, we give an overview of the theory behind pyFFS.
This can be useful to those who wish to understand what is happening under the hood and sheds light as to when using FS coefficients over DFT ones may be of interest.

\subsection{Numerically compute FS coefficients}

At first glance, it seems unlikely that FS coefficients could be of practical use in computational scenarios, as their very definition  necessitates the integral of a continuous function~\cite{vetterli2014foundations}
\begin{equation}\label{eq:fs}
X_{k}^{\text{FS}} = \frac{1}{T} \int_{T_{c} - \frac{T}{2}}^{T_{c} + \frac{T}{2}} x(t) \exp\bigParen{-j \frac{2 \pi}{T} k t} dt,
\end{equation}
where $x : \bR \to \bC$ is a $T$-periodic function and $T_{c}$ is any period mid-point.

The synthesis equation, on the other hand, expresses this function in terms of a discrete, albeit infinite, set of samples, namely the FS coefficients
\begin{equation}\label{eq:ifs}
x(t) = \sum_{k \in \bZ} X_{k}^{\text{FS}} \exp\bigParen{j \frac{2 \pi}{T} k t}.
\end{equation}

For a bandlimited signal, we can write~\cref{eq:ifs} as
\begin{equation}
\label{eq:bl_ifs}
x(t) = \sum_{k=-N}^{N} X_{k}^{\text{FS}} \exp\bigParen{j \frac{2 \pi}{T} k t}, \quad N \ge 0,
\end{equation}
where $ x(t) $ is said to have a bandwidth of $N_{\text{FS}} = 2 N + 1$.\footnote{This is equivalent to a maximal frequency of $ N / T $ Hz.} Moreover, by taking uniform samples with a sampling period of $ T_s = T/N_{\text{FS}} $, we obtain
\begin{align}
\label{eq:bl_ifs_samp}
x(nT_s) = \bigBrack{\bbx}_{n} &= \sum_{k = -N}^N X_{k}^{\text{FS}} \exp\bigParen{j \frac{2 \pi}{T} k (nT_s)} \nonumber
\\ & \overset{(a)}{=} \sum_{k = -N}^N X_{k}^{\text{FS}}  \exp\bigParen{j \frac{2 \pi}{N_{\text{FS}}} kn} \nonumber
\\ & \overset{(b)}{=} \sum_{k = -N}^N X_{k}^{\text{FS}}  W_{ N_{\text{FS}} }^{-n k},
\end{align} 
where $ (a) $ uses $ (T_s / T) = 1 / N_{\text{FS}}$ and $ (b) $ uses $W_{N} = \exp\bigParen{-j \frac{2\pi}{N}}$. Uniform sampling is equivalent to convolution with a Dirac stream in the FS coefficient domain. Therefore the sequence $\bigCurly{X_{k}^{\text{FS}} \in \bC, \; k \in \bZ}$ is $ N_{\text{FS}} $-periodic, as shown in \Cref{fig:sample_fs_coeff}. Moreover, choosing $ T_s \leq T/N_{\text{FS}} $ ensures that the FS coefficients remain ``intact'', i.e.\ there is no distortion due to aliasing from the overlapping spectrum replicas.

In~\cref{eq:bl_ifs_samp} we have an expression that is equivalent to what is commonly referred to as the discrete Fourier series (DFS)~\cite{Prandoni2008}
\begin{equation}\label{eq:idfs}
\bigBrack{\bbx}_{n} = \sum_{k = 0}^{N_{\text{FS}}-1} \bigBrack{\bbX}_{k}  \exp\bigParen{j \frac{2 \pi}{N_{\text{FS}}} k n}, \quad n \in \mathbb{Z},
\end{equation}
which is similar to the inverse discrete Fourier transform (IDFT), except that the time index $ n $ spans all integers. 
As the coefficients of the sampled function are $ N_{\text{FS}} $-periodic, shifting the summation in~\cref{eq:bl_ifs_samp} to $ [0, N_{\text{FS}}-1] $, as in~\cref{eq:idfs},  is equivalent to the summation over $ [-N, N] $. 

\begin{figure}[t!]
	\centering
	\includegraphics[width=0.99\linewidth]{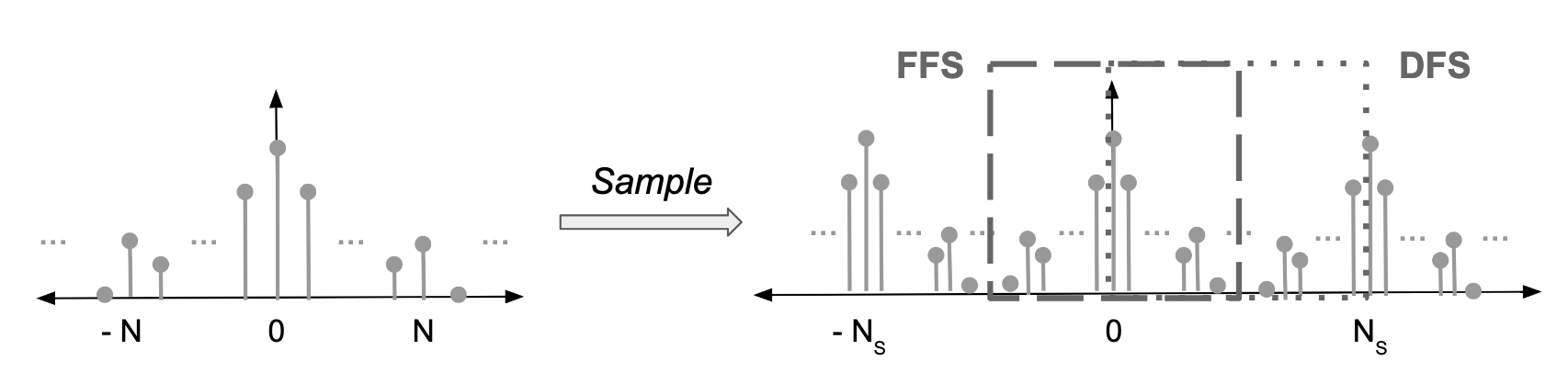} 
	\caption{Visualizing the Fourier series (FS) coefficients of a sampled, bandlimited signal: (left) original FS coefficients of a signal whose bandwidth is $ N_{\text{FS}} = 2N + 1 $; (right) sampled FS coefficients where $ N_{\text{s}} \geq N_{\text{FS}} $ such that the sampling period is $ T_{\text{s}} = T / N_{\text{s}} $. When performing synthesis, i.e.\ interpolating coefficients for time-domain values, our formulation~\cref{eq:bl_ifs_samp} sums over coefficients centered around zero (dashed box), while the discrete Fourier series~\cref{eq:idfs} sums over positive coefficients (dotted box). This is equivalent as the sampled Fourier series coefficients are a periodic sequence.}
	\label{fig:sample_fs_coeff}
\end{figure}


With~\cref {eq:bl_ifs_samp}, we can numerically compute the FS coefficients of a bandlimited, periodic signal by setting up a system of equations with at least as many samples as FS coefficients:
\[
\begin{bmatrix}
\tiny
\bigBrack{\bbx}_{-N} \\
\bigBrack{\bbx}_{-N+1}  \\
\vdots \\
\bigBrack{\bbx}_{N} 
\end{bmatrix}
= 
\scalebox{0.9}{
	$\begin{bmatrix} 
	W_{N_{FS}}^{-(-N)(-N)} & W_{N_{FS}}^{-(-N)(-N+1)}  & \cdots & W_{N_{FS}}^{-(-N)(N)} \\[5pt]
	W_{N_{FS}}^{-(-N+1)(-N)}  & W_{N_{FS}}^{-(-N+1)(-N+1)}  & \cdots & W_{N_{FS}}^{-(-N+1)(N)} \\
	\vdots & \vdots & \ddots & \vdots  \\[5pt]
	W_{N_{FS}}^{-(-N)(N)} & W_{N_{FS}}^{-(-N+1)(N)}  & \cdots & W_{N_{FS}}^{-(N)(N)} \\
	\end{bmatrix}$}
\begin{bmatrix}
X_{-N}^{\text{FS}}\\
X^{\text{FS}}_{-N+1}  \\
\vdots \\
X^{\text{FS}}_{N}
\end{bmatrix}.
\]

A naive approach could solve this system of equations for the FS coefficients by multiplying the left-hand side with the inverse of the matrix, requiring $ \mathcal{O}(N_{FS}^3) $ operations. However, the above matrix is unitary, allowing us to multiply both sides by the complex conjugate of the matrix rather than having to invert it. This gives a complexity of $ \mathcal{O}(N_{FS}^2) $. With some manipulation, the matrix-vector relationship can be turned into one where the matrix is the DFT matrix. We can then exploit the FFT algorithm to further reduce the complexity to $ \mathcal{O}(N_{FS} \log N_{FS}) $~\cite{Cooley1965}. This result is given in the following theorem, with the proof provided in \Cref{sec:proofs}.

\begin{theorem}[Fast Fourier series]\label{thm:ffs_odd}
	Let $x : \bR \to \bC$ be a $T$-periodic function of bandwidth $N_{\text{FS}} = 2 N +
	1$, with $T_{c} \in \bR$ the mid-point of any period.  Let the zero-padding amount $Q \in 2 \bN$ be an arbitrary
	even integer such that the number of samples $N_{\text{s}} = N_{\text{FS}} + Q $. Then
	\begin{gather}
	\bbx = N_{\text{s}} \iDFT_{N_s}\bigParen{\bbX^{\text{FS}} \odot B_{1}^{\bbE_{1}}} \odot B_{2}^{N \bbE_{2}}, \label{eq:ffs_odd_phi} \\
	\bbX^{\text{FS}} = \frac{1}{N_{\text{s}}} \DFT_{N_s}\bigParen{\bbx \odot B_{2}^{-N \bbE_{2}} } \odot B_{1}^{-\bbE_{1}}, 
	\label{eq:ffs_odd_phiFS}
	\end{gather}
	where
	\begin{gather}
	\bbx = \bigBrack{\bbx(t_{0}), \ldots, \bbx(t_{M}), \bbx(t_{-M}), \ldots, \bbx(t_{-1})} \in \bC^{N_{\text{s}}}, \label{eq:ffs_input} \\
	t_{n} = T_{c} + \frac{T}{N_{\text{s}}} n, \quad n \in \bZ, \label{eq:ffs_odd_tn}\\
	M = (N_{\text{s}} - 1) / 2, \nonumber\\
	\bbX^{\text{FS}} = \bigBrack{X_{-N}^{\text{FS}}, \ldots, X_{N}^{\text{FS}}, \bbZero_{Q}} \in \bC^{N_{\text{s}}}, \nonumber
	\end{gather}
	and
	\begin{minipage}{0.47\textwidth}
		\begin{center}
			\begin{gather*}
			B_{1} = \exp\bigParen{j \frac{2 \pi}{T} T_{c}} \in \bC, \\
			\bbE_{1} = \bigBrack{-N, \ldots, N, \bbZero_{Q}} \in \bZ^{N_{\text{s}}},
			\end{gather*}
		\end{center}
	\end{minipage}
	\begin{minipage}{0.47\textwidth}
		\begin{center}
			\begin{gather*}
			B_{2} = \exp\bigParen{-j \frac{2 \pi}{N_{\text{s}}}} \in \bC, \\
			\bbE_{2} = \bigBrack{0, \ldots, M, -M, \ldots, -1} \in \bZ^{N_{\text{s}}}.
			\end{gather*}
		\end{center}
	\end{minipage}
\end{theorem}
$ \DFT_N $ and $ \iDFT_N $ denote the length-$ N $ DFT and IDFT respectively. The operation $ \odot $ is an element-wise multiplication, with which in~\cref{eq:ffs_odd_phi,eq:ffs_odd_phiFS} we modulate the input and the output of the standard DFT and IDFT operations. This modulation shifts the summation bounds in~\cref{eq:bl_ifs_samp} to those of the DFT, allowing direct use of the FFT for an efficient computation. The zero-padding $ Q $ can be used to set the FFT length to a highly composite value for faster computation. 

The above theorem assumes an odd-length $ N_{\text{s}} $. For an even-length sequence, a slight modification of the sample locations $ t_n $ and modulation terms $ B_1, B_2 $ is necessary. As the general idea is the same, we refer the reader to \Cref{sec:proofs} for the presentation and proof for the odd- and the even-length cases.

\subsection{Efficient interpolation of FS coefficients}

There is nothing particularly unique about the FS coefficients given by \Cref{thm:ffs_odd} in~\cref {eq:ffs_odd_phiFS}. They are essentially DFT coefficients modulated, reordered, and scaled so that the resulting sequence is directly the FS coefficients within $ [-N, N] $. However, their intrinsic link to a continuous signal makes them all the more convenient when we consider interpolation at arbitrary sample locations, rather than griding or sub-griding at the discrete level. From a practitioner's perspective, it may also be more intuitive to think and work directly with the continuous-domain sample locations rather than the corresponding sample (or subsample) values. 

\Cref{eq:bl_ifs} could be used to compute the value of the underlying bandlimited function at arbitrary sample locations, requiring $\cO(MN_{\text{FS}})$ operations for $ M $ sample points. Below we show one way these FS coefficients can be interpolated at $M$ regularly-spaced samples in a more efficient manner. This is done by making use of the chirp Z-transform (CZT).

\begin{definition}[Chirp Z-transform]\label{def:czt}
	Let $\bbx \in \bC^{N}$. The length-$M$ \emph{chirp Z-transform} $\CZT_{N}^{M}(\bbx) \in
	\bC^{M}$ of parameters $A, W \in \bC^{*}$ is defined as
	\cite{rabiner1969chirp}
	\begin{equation}\label{eq:czt}
	\bigBrack{\CZT_{N}^{M}(\bbx)}_{k} = \sum_{n = 0}^{N - 1} \bigBrack{\bbx}_{n} A^{-n} W^{n k}, \quad k \in \bigCurly{0, \ldots, M - 1},
	\end{equation}
where $ A$ is the complex starting point and $W$ is the complex ratio between points along a logarithmic spiral contour. The CZT is a generalization of the DFT which samples the $Z$ plane at uniformly-spaced points along the unit circle.
\end{definition}
$\CZT_{N}^{M}$ can be efficiently computed using the $\DFT$ and $\iDFT$ in $\cO(L \log L)$
operations, where $L \ge N + M - 1$. This is done via Bluestein's algorithm~\cite{Bluestein1970}. A similar formulation of the CZT and its efficient computation is known as the fast fractional FT algorithm~\cite{Bailey1991}.

The theorem below makes use of the CZT and Bluestein's algorithm to perform efficient interpolation at regularly-spaced values. 

\begin{theorem}[FS interpolation, $x : \bR \to \bC$]\label{thm:fast_interp_complex}
	Let $x: \bR \to \bC$ be a $T$-periodic function of bandwidth $N_{\text{FS}}= 2 N +
	1$.  Let $a < b \in \bR$ be the end-points of an interval on which we want to evaluate
	$M$ equi-spaced samples of $x$. Then
	\begin{equation}\label{eq:fast_interp_complex}
	\bbx = A^{N} \CZT_{N_{\text{FS}}}^{M}(\bbX^{\text{FS}}) \odot W^{-N \bbE},
	\end{equation}
	where
	
	\begin{minipage}{0.45\textwidth}
		\begin{center}
			\begin{gather*}
			\bbx = \bigBrack{x(t_{0}), \ldots, x(t_{M-1})} \in \bC^{M}, \\
			\bbX^{\text{FS}} = \bigBrack{X_{-N}^{\text{FS}}, \ldots, X_{N}^{\text{FS}}} \in \bC^{N_{\text{FS}}}, \\
			t_{n} = a + \frac{b - a}{M - 1} n, \quad n \in \bZ,
			\end{gather*}
		\end{center}
	\end{minipage}
	\begin{minipage}{0.45\textwidth}
		\begin{center}
			\begin{gather*}
			A = \exp\bigParen{-j \frac{2\pi}{T} a}, \\
			W = \exp\bigParen{j \frac{2 \pi}{T} \frac{b - a}{M - 1}}, \\
			\bbE = \bigBrack{0, \ldots, M -1} \in \bN^{M},
			\end{gather*}
		\end{center}
	\end{minipage}

and $\CZT$ (of parameters $A, W$) is as defined in \cref{def:czt}.
\end{theorem}
The proof can be found in \Cref{sec:fs_interp_proof}. A similar interpolation technique with the fractional FT is presented in~\cite{Bailey1991}.

With \Cref{thm:fast_interp_complex}, one can interpolate sub-sections of a period efficiently. Moreover, it is possible to perform DFTs of a smaller length than what would be normally required with a more standard IDFT interpolation approach, namely zero-padding the DFT coefficients and taking a longer IDFT for an increase in temporal resolution across the \emph{entire} period. 

\subsubsection*{Complexity comparison}

To get an idea of when it is preferable to use the proposed technique over zero-padding DFT coefficients, we present a rough complexity analysis. For a $ T $-periodic function that we would like to evaluate at steps of $ \Delta t $, we would need $ N_{\text{target}} =  \ceil{T/\Delta t} $ samples within a single period. If we had $ N < N_{\text{target}} $ samples,\footnote{With $ N > N_{\text{FS}} $ so that we can have ideal reconstruction according to the Nyquist-Shannon sampling theorem.} we would need to pad the DFT coefficients with $ (N_\text{target} - N) $ zeros in order to get this temporal resolution, resulting in an IDFT with complexity $ \cO ( N_\text{target} \log N_\text{target} ) $. With \Cref{thm:fast_interp_complex}, interpolating at steps of $ \Delta t $ over the entire period $ T $  would lend to a computational complexity of $ \cO \big( (N_\text{target} + N_{\text{FS}}) \log (N_\text{target}+N_{\text{FS}}) \big) $, which is certainly not advantageous. The benefits arise when we wish to interpolate over a smaller region within the period, i.e.\ when zooming in on a section of length $ mT $ with $ m \in (0, 1) $. In such a scenario, the number of interpolation points is $ M = \ceil{mT/\Delta t}  $. As the complexity of \Cref{thm:fast_interp_complex} is $ \cO \big( (M + N_{\text{FS}}) \log (M+N_{\text{FS}}) \big) $, it can be more efficient than zero-padding the DFT coefficients when $ (M + N_{\text{FS}} )$ is less than the padded IDFT length $N_{\text{target}} = \ceil{T/\Delta t}  $.

Both the proposed approach in \Cref{thm:fast_interp_complex} and interpolation by zero-padding the DFT coefficients are bandlimited interpolation techniques. Given all the FS or DFT coefficients of a bandlimited signal, both yield a distortionless interpolation. A comparison between these two approaches is done in \Cref{sec:bench_interp}.

Another approach for bandlimited interpolation is sinc interpolation with Dirichlet apodization~\cite{Thevenaz2009}, and it allows one to focus on a specific region much like  \Cref{thm:fast_interp_complex}. However, its complexity is given by $ \cO ( M  N_{\text{FS}}) $, which quickly becomes prohibitive as the number of interpolation points or the bandwidth increases. So we do not consider this approach in our comparison.

\subsection{When to use FS coefficients}

In our discussion above, we already came across two requirements on the input in order to numerically compute its FS coefficients from a discrete set of samples without distortion: 
\begin{itemize}
	\item \emph{Periodic} as the FS coefficients are defined for such signals.
	\item \emph{Bandlimited} so that we can perfectly recover the FS coefficients from a finite set of discrete samples. 
\end{itemize}
In fact, the DFT makes similar assumptions on periodicity. As we saw with the DFS, when we consider indices outside of the sequence's finite support, we observe a periodic structure. Consequently, when one wishes to interpolate outside the support of a finite sequence, it is common to apply a tapering window in order to avoid sharp discontinuities at the boundaries. Similarly, in order to apply \Cref{thm:ffs_odd}, we can consider the samples of an arbitrary continuous function as a truncation that we periodize and optionally taper to remove any discontinuities at the boundaries.
For example, in \Cref{fig:discontinuous_sine} we compute the FS coefficients of a truncated sinusoid with a large discontinuity at the border. Applying a tapering window attenuates the boundaries to zero, thereby removing the discontinuity as seen in the periodization (bottom of \Cref{fig:discontinuous_sine_time}) and suppressing higher frequency coefficients that arise due to this discontinuity (\Cref{fig:discontinuous_sine_fs}).

\begin{figure}[t!]
	\centering
	\begin{subfigure}{.49\textwidth}
		\centering
		\includegraphics[width=0.99\linewidth]{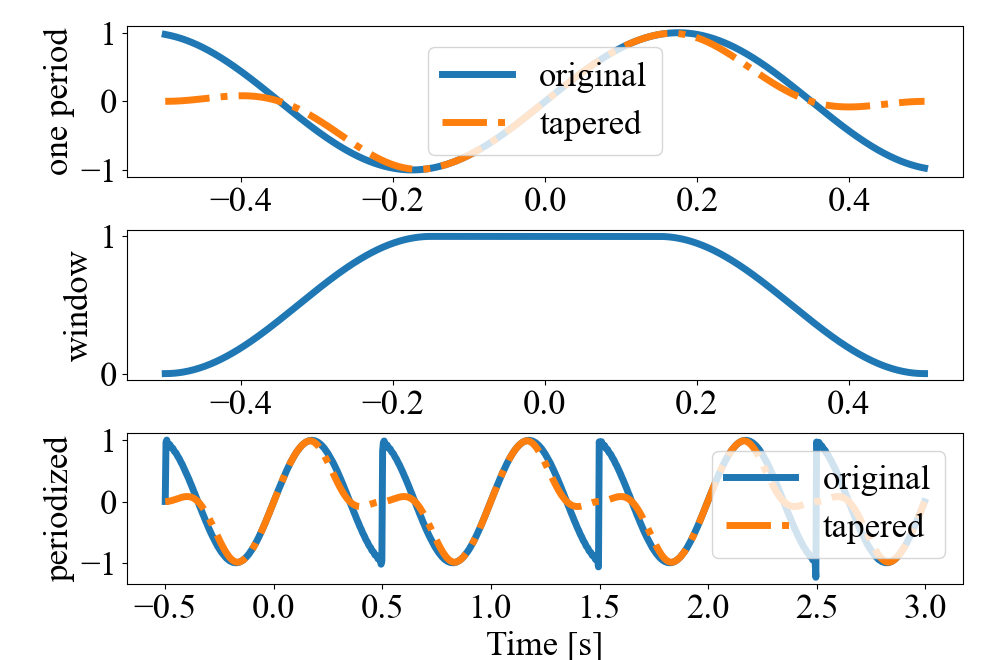} 
		\caption{}
		\label{fig:discontinuous_sine_time}
	\end{subfigure}
	\hfill
	\begin{subfigure}{.49\textwidth}
		\centering
		\includegraphics[width=0.99\linewidth]{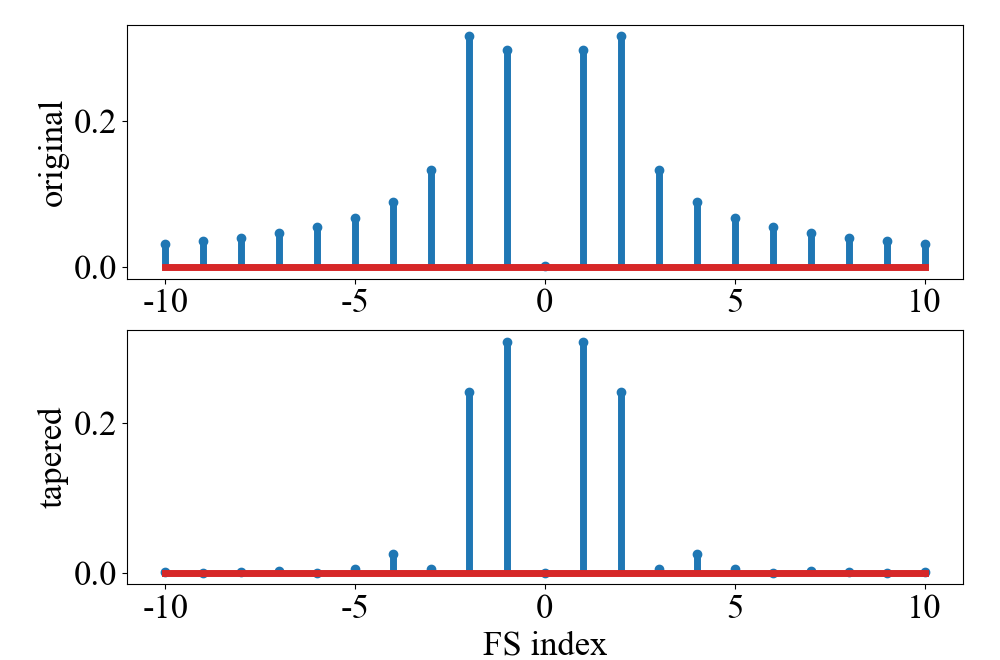}
		\caption{}
		\label{fig:discontinuous_sine_fs}
	\end{subfigure}
	\caption{Example to motivate tapering window before Fourier analysis. Notice in~\cref{fig:discontinuous_sine_fs}, the higher frequency coefficients that appear in the Fourier analysis due to the discontinuity (top), and how they are suppressed when applying a tapering window (below).}
	\label{fig:discontinuous_sine}
\end{figure}

Although not necessarily required for the DFT, bandlimitness of the input sequence is desired prior to sampling so that distortion due to aliasing is minimized. However, after sampling there is not much that can be done to remove such aliasing. The same is true for computing the FS coefficients. However, if one has control on the acquisition process, the sampling rate can be set to avoid spectral overlap or the signal can be bandlimited appropriately. As we saw in \cref{eq:bl_ifs}, bandlimiting corresponds to truncating the FS summation.

%
%


FS coefficients coupled with knowledge about the underlying continuous signal therefore allows us to better understand, and perhaps control, where errors could arise during our analysis and processing:
\begin{itemize}
	\item From aliasing if the underlying function could not be bandlimited or sampled appropriately.
	\item From truncation of the FS coefficients to ensure bandlimitedness. The mean-squared error of this truncation can be bounded~\cite{Giardina1972}.
	\item From a tapering window to remove border discontinuities if a function was simply periodized by truncation and repetition.
\end{itemize}
Applying the DFT has similar constraints and consequences that tend to get ignored in practice, as they may not be critical for the application or cannot be corrected with the discrete samples alone.

In practice, discrete samples are bandlimited due to the sampling operation, although there might be aliasing if the proper filter could not be used beforehand, e.g.\ in medical imaging~\cite{Thevenaz2009}. In such scenarios, we can only expect to obtain the FS coefficients of the underlying aliased signal.


\subsection{Multidimensional}

The above methods are not limited to a single dimension. As the DFT plays a central role, the FS coefficient and interpolation tools can be extended to multidimensional in a similar fashion as the DFT.
By leveraging the property that the multidimensional DFT can be computed as a composition of one-dimensional DFTs along each dimension, the FS computation and interpolation methods above, namely \Cref{thm:ffs_odd,thm:fast_interp_complex}, can be similarly computed along each dimension of the samples and the FS coefficients respectively.

This application of the 1-D FFT along each dimension results in a complexity of $\cO(N \log N)$, where $N = N_1 \cdot N_2 \cdot \cdots N_D$ for a general N-D signal. More efficient algorithms for performing the multidimensional DFT do exist~\cite{Duhamel1990}, which we make use of in our implementations of \Cref{thm:ffs_odd,thm:fast_interp_complex} in multiple dimensions.

%% file: content/3_pyffs.tex
\section{pyFFS overview and usage}
\label{sec:pyffs}

pyFFS is a Python library for performing efficient and distortionless FS coefficient computation, convolution, and interpolation for periodic, bandlimited signals. The goal is to provide an intuitive tool to numerically work with such signals of any dimension $ D $. Just as the FFT functions from NumPy~\cite{Harris2020} and SciPy~\cite{Virtanen2020} can be used without too much thought about the internal details, we have created a user interface for FS computations that can also be used out-of-the-box for the appropriate scenario.

\subsection{Fourier series analysis and synthesis}

The user interface for 1-D functions is shown below. Note that the samples provided to \texttt{ffs} must be in the same order as specified in \cref{eq:ffs_input}, which is not in chronological order. The method \texttt{ffs\textunderscore sample} returns the timestamps and indices necessary for ensuring the samples provided to \texttt{ffs} are in the expected order.

\begin{lstlisting}[language=Python, caption=1-D fast Fourier series analysis and synthesis.]
# determine appropriate timestamps and indices for rearranging input
sample_points, idx = pyffs.ffs_sample(
  T,      # function period
  N_FS,   # function bandwidth, i.e. number of FS coefficients (odd)
  T_c,    # function center
  N_s     # number of samples
)

# sample a known function at the correctly ordered timestamps
x = pyffs.func.dirichlet(sample_points, T, T_c, N_FS)
# OR rearrange ordered samples using `idx`
# x = x[idx]

# compute FS coefficients 
x_FS = pyffs.ffs(x, T, T_c, N_FS)

# back to samples with inverse transform
x_r = pyffs.iffs(x_FS, T, T_c, N_FS)    # equivalent to x
\end{lstlisting}
The user interface for the general N-D case is shown below, with the specific example of 2-D. As in the 1-D case, samples provided to \texttt{ffsn} are not in increasing order of the input variables. The method \texttt{ffsn\textunderscore sample} returns the locations and indices necessary for making sure the samples provided to \texttt{ffsn} are in the expected order. Alternatively, the method \texttt{ffs\textunderscore shift} can be used to reorder the samples.

\begin{lstlisting}[language=Python, caption=2-D fast Fourier series analysis and synthesis.]
T = [T_x, T_y]         # list of periods for each dimension
T_c = [T_cx, T_cy]     # list of function centers for each dimension
N_FS = [N_FSx, N_FSy]  # list of function bandwidths for each dimension
N_s = [N_sx, N_sy]     # number of samples per dimension

# determine appropriate timestamps and indices for rearranging input
sample_points, idx = pyffs.ffsn_sample(T=T, N_FS=N_FS, T_c=T_c, N_s=N_s)

# sample a known function at the correctly ordered timestamps
x = pyffs.func.dirichlet_2D(sample_points, T, T_c, N_FS)
# OR rearrange ordered samples
# x = pyffs.ffs_shift(x)

# compute FS coefficients 
x_FS = pyffs.ffsn(x, T=T, T_c=T_c, N_FS=N_FS)

# go back to samples
x_r = pyffs.iffsn(x_FS, T=T, T_c=T_c, N_FS=N_FS)    # equivalent to x
\end{lstlisting}

\subsection{Circular convolution}

The user interface for N-D functions is shown below. \texttt{f} and \texttt{h} are function values at the sampling points specified by \texttt{ffsn\textunderscore sample}, namely they must be functions of the same period \texttt{T}, with the same period center \texttt{T\textunderscore c}, and of the same bandwidth \texttt{N\textunderscore FS}.

Samples can be provided in their natural order or in the order expected by \texttt{ffsn}. By default, the argument \texttt{reorder} is set to \texttt{True}, such that samples are expected in their natural order and are reordered internally.
The output samples are returned in the same order as the inputs.

\begin{lstlisting}[language=Python, caption=Circular convolution of two N-D functions through Fourier series coefficients.]
out = pyffs.convolve(
  f=f,  # samples of one function in the convolution
  h=h,  # samples of the other function in the convolution
  T=T,  # period(s) of both functions along all dimensions
  T_c=T_c,  # period center(s) of both functions along all dimensions
  N_FS=N_FS,   # number of FS coefficients for both functions along all dimensions
  reorder=True  # whether input samples should be reordered into expected order for ffsn
)
\end{lstlisting}


\subsection{Interpolation}

The user interface for 1-D functions is shown below.

\begin{lstlisting}[language=Python, caption=1-D fast Fourier series interpolation.]
x_interp = pyffs.fs_interp(
  x_FS,   # FS coefficients in increasing order of index
  T,      # period
  a,      # start time
  b,      # stop stop
  M       # number of points
)

\end{lstlisting}
The user interface for the general N-D case is shown below, with the specific example of 2-D.

\begin{lstlisting}[language=Python, caption=2-D fast Fourier series interpolation.]
x_interp = pyffs.fs_interpn(
  x_FS,            # multidimensional FS coefficients
  T=[T_x, T_y],    # list of periods for each dimension
  a=[a_x, a_y],    # list of start points for each dimension
  b=[b_x, b_y],    # list of stop points for each dimension
  M=[M_x, M_y]     # number of samples per dimension
)
\end{lstlisting}
In both cases, the provided FS coefficients must be ordered such that the indices are in increasing order, as returned by \texttt{ffs} and \texttt{ffsn}.

\subsection{GPU support}

GPU usage can lead to a significant reduction in computation time if a task consists of many operations that can be done in parallel. The FFT algorithm can be parallelized and could therefore benefit from such a reduction in computation time. As the fast FS algorithm presented in \cref{thm:ffs_odd} makes use of the DFT and IDFT, it can directly benefit from this speed-up, and so can the interpolation of FS coefficients described in \cref{thm:fast_interp_complex}, as Bluestein's algorithm for the CZT employs the DFT and IDFT.

GPU support is available through the CuPy library~\cite{Okuta2017}. If the appropriate version of CuPy is installed,\footnote{See installation guide: \rurl{https://docs.cupy.dev/en/stable/install.html}} nearly all array operations will take place on the GPU if the provided input is a CuPy array, as shown below. NumPy arrays can be passed if one wishes to still perform operations on the CPU.

\begin{lstlisting}[language=Python, caption=GPU support through CuPy.]
import cupy as cp

x_cp = cp.array(x)   # convert existing `numpy` array to `cupy` array

# apply functions like before, array operations take place on GPU
x_FS = pyffs.ffs(x_cp, T, T_c, N_FS)    # compute FS coefficients
x_r = pyffs.iffs(x_FS, T, T_c, N_FS)    # back to samples
y = pyffs.convolve(x_cp, x_cp, T, T_c, N_FS)    # convolve
x_interp = pyffs.fs_interp(x_FS, T, a, b, M)  # interpolate

\end{lstlisting}
Note that converting between CuPy and NumPy requires data transfer between the CPU and GPU, which could be costly for large arrays. Therefore, if passing CuPy arrays to pyFFS, it is recommended to perform as much pre-processing and post-processing as possible on the GPU in order to limit such data transfer.

\subsection{Summary}

\Cref{tab:func} compares equivalent functions between pyFFS and SciPy.
For a more extensive documentation and the latest information, we refer to \rurl{pyffs.readthedocs.io}, and example scripts can be found in the \texttt{examples} folder of the repository:  \rurl{github.com/imagingofthings/pyFFS/tree/master/examples}.

\begin{table}[t!]
		\centering
		\resizebox{\columnwidth}{!}{
	\begin{tabular}{|c|c|c|}
		\hline  & \textbf{pyFFS} & \textbf{SciPy} \\ 
		\hline 1-D Fourier analysis & \texttt{pyffs.ffs} & \texttt{scipy.fft.fft} \\ 
		\hline 1-D Fourier synthesis &  \texttt{pyffs.iffs} & \texttt{scipy.fft.ifft} \\ 
		\hline N-D Fourier analysis & \texttt{pyffs.ffsn}  & \texttt{scipy.fft.fftn} \\ 
		\hline N-D Fourier synthesis & \texttt{pyffs.iffsn} &  \texttt{scipy.fft.ifftn} \\
		\hline N-D convolution & \texttt{pyffs.convolve} & \texttt{scipy.signal.fftconvolve} \\ 
		\hline 1-D bandlimited interpolation & \texttt{pyffs.fs\_interp} & \texttt{scipy.signal.resample} \\
		\hline N-D bandlimited interpolation & \texttt{pyffs.fs\_interpn} & -\\ 
		\hline 
	\end{tabular} 
    }
	\caption{Functionality comparison between pyFFS and SciPy~\cite{Virtanen2020}. SciPy's convolution zero-pads the inputs in order to approximate a linear convolution, while pyFFS performs a circular convolution. Within SciPy, circular convolution is only supported for 2-D by calling \texttt{scipy.signal.convolve2d}  with the parameter \texttt{boundary=`wrap'}.
	For N-D bandlimited interpolation with SciPy, it is possible to use \texttt{scipy.signal.resample} along each dimension. However, there is no one-shot function.}
	\label{tab:func}
\end{table}

%% file: content/4_benchmark.tex
\section{Benchmarking} 


\label{sec:benchmark} 
Computational efficiency is a primary objective for pyFFS. 
In this section, we present several benchmarking results to compare the computational speed between pyFFS and SciPy for convolution and interpolation, and to demonstrate the benefits of GPU acceleration.\footnote{The scripts to reproduce these results can be found in the in the \texttt{profile} folder of the repository:  \rurl{github.com/imagingofthings/pyFFS/tree/master/profile}.}
All benchmarking is performed on a Lenovo ThinkPad P15 Gen 1 laptop, with an Intel i7-10850H six-core processor and an NVIDIA Quadro RTX 3000 GPU (when applicable).

\subsection{Convolution}


Before presenting the benchmarking results, we show a toy example of 1-D convolution, to compare the outputs of pyFFS and SciPy.

\begin{figure}[t!]
	\centering
	\begin{subfigure}{.49\textwidth}
	\includegraphics[width=0.99\linewidth]{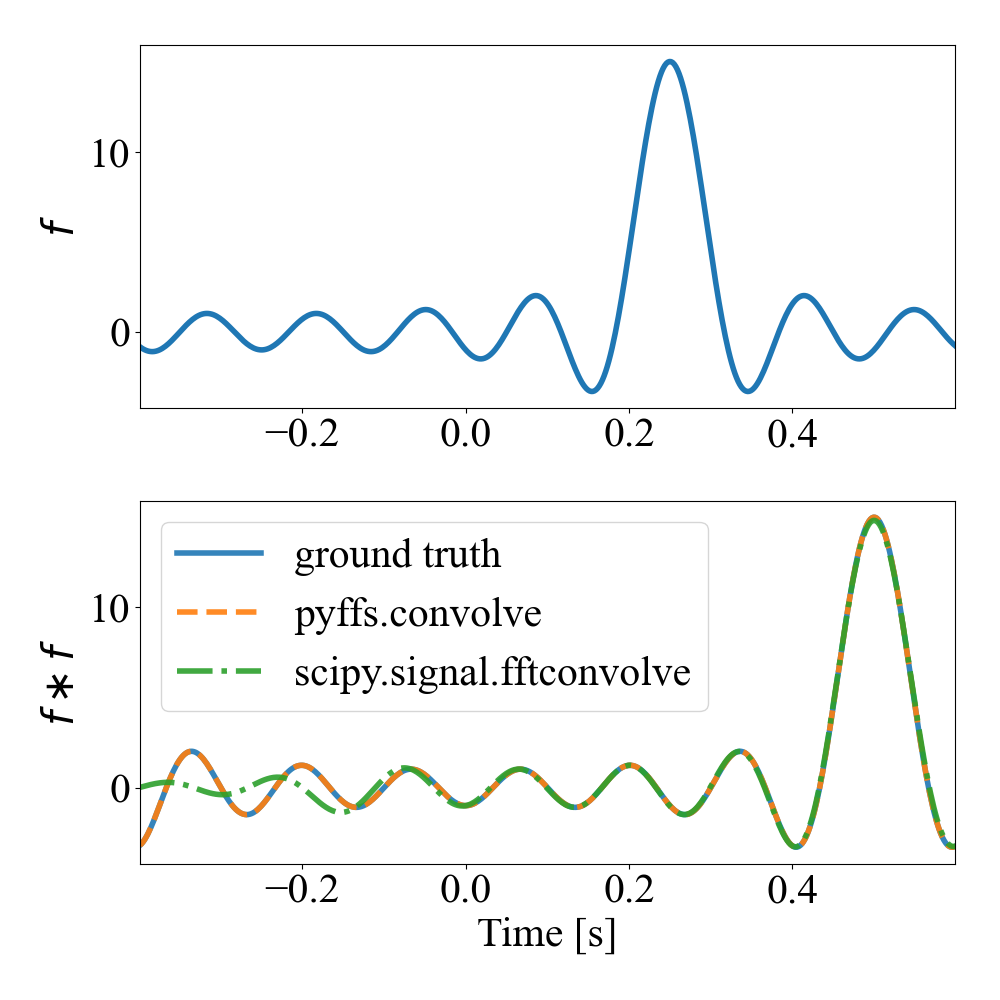} 
	\caption{}
	\label{fig:convolve_1d}
	\end{subfigure}
	\hfill
	\begin{subfigure}{.49\textwidth}
	\centering
	\includegraphics[width=0.99\linewidth]{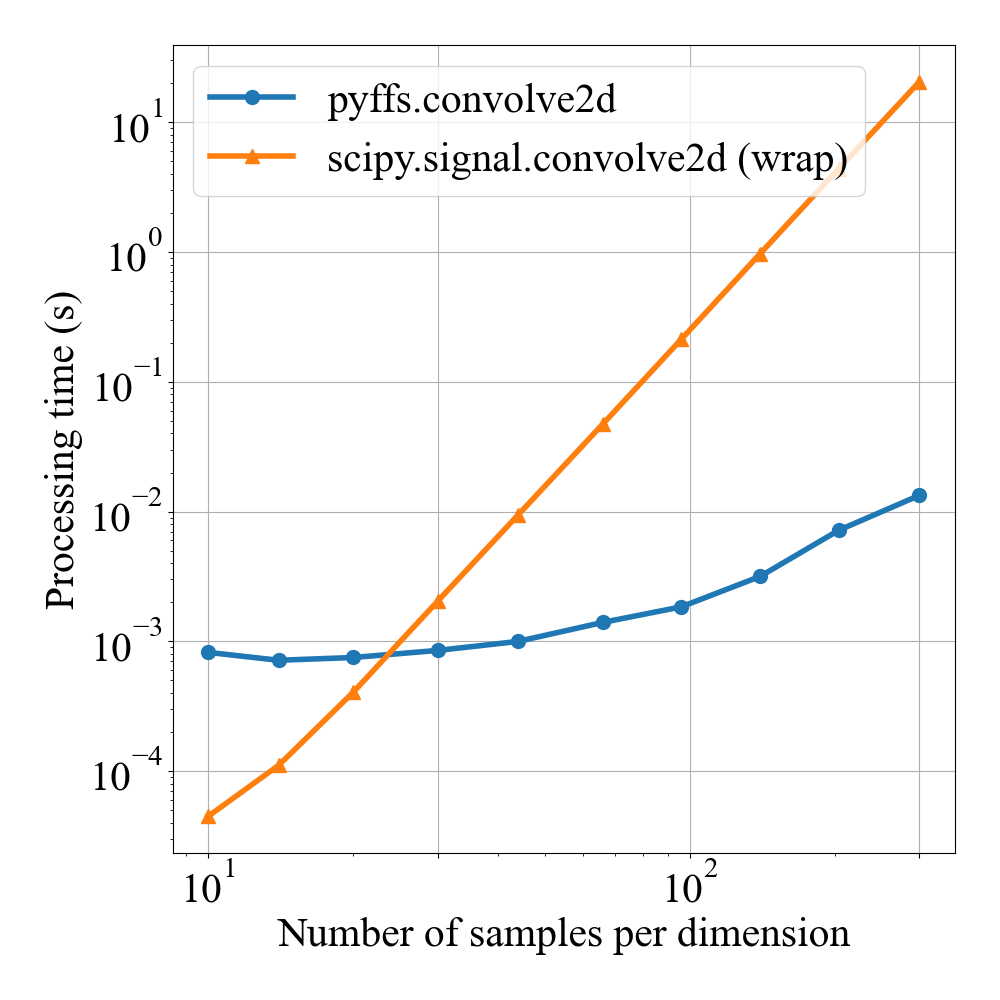} 
	\caption{}
	\label{fig:profile_convolve2d}
\end{subfigure}
\caption{(a) Comparing 1-D convolution between pyFFS and SciPy. pyFFS performs a circular convolution as is expected by Fourier series coefficients, while SciPy's \texttt{fftconvolve} approximates a linear convolution by zero-padding DFT coefficients. (b) Profiling 2-D bandlimited circular convolution between pyFFS and SciPy.}
\end{figure}

Both DFT and FS analysis assume periodic functions, discrete and continuous ones respectively~\cite{vetterli2014foundations}. Moreover, the convolution of two periodic functions of the same period results in a circular convolution, such that the output is also periodic. In certain scenarios, one may be interested in the convolution of two signals that are not necessarily periodic, e.g.\ a speech recording and a room impulse response. By zero-padding the two inputs, a linear convolution can be computed from the circular convolution of two discrete inputs. This is precisely what the convolve function of SciPy does. However, SciPy offers no function for the 1-D circular convolution.\footnote{It can be manually done by taking the inverse DFT of the product of the DFTs of the two inputs.} pyFFS' convolution function performs the circular convolution, as expected by FS convolution theory. This explains the noticeable difference between the two approaches in \Cref{fig:convolve_1d}, in particular at the leftmost boundary. The SciPy approach tapers off due to zero-padding as it approximates a linear convolution, while pyFFS performs a faithful circular convolution.


\subsubsection*{Benchmark}

%

As SciPy's function for 1-D convolution yields a different result, we do not benchmark pyFFS against it. For 2-D convolution, SciPy has a function \texttt{scipy.signal.convolve2d} that can perform a circular convolution when the argument \texttt{boundary=`wrap'}. In \Cref{fig:profile_convolve2d} we compare \texttt{pyffs.convolve2d} against \texttt{scipy.signal.convolve2d}, as these two lead to the same output. As the number of samples per dimension grows, using pyFFS is noticeably faster than the equivalent function in SciPy. At around $ 100 $ samples per dimension, pyFFS is already two orders of magnitude faster than SciPy.

\subsection{Interpolation}
\label{sec:bench_interp}

Likewise, before presenting the benchmarking results, we show toy examples of interpolation in 1-D and 2-D, comparing \Cref{thm:fast_interp_complex} and interpolation by zero-padding DFT coefficients.

\Cref{fig:interp1d_input,fig:interp1d_comparison} show the interpolation of a section of a 1-D Dirichlet function
\begin{align}
\label{eq:dirichlet}
\phi(t) = \sum_{k=-N}^{N}\exp\Big(j \frac{2\pi}{T} k(t - T_c) \Big),
\end{align}
whose bandwidth is given by $ N_{\text{FS}} = 2N + 1 $. \Cref{thm:fast_interp_complex}, which makes use of \texttt{pyffs.fs\textunderscore interp}, and interpolating zero-padded DFT coefficients, which employs \texttt{scipy.signal.resample}, perfectly match the ground truth function, as shown in \Cref{fig:interp1d_comparison}. This is expected as both approaches are bandlimited interpolation techniques and the target function is bandlimited.

\begin{figure}[t!]
	\centering
	\begin{subfigure}{.49\textwidth}
		\centering
		\includegraphics[width=0.99\linewidth]{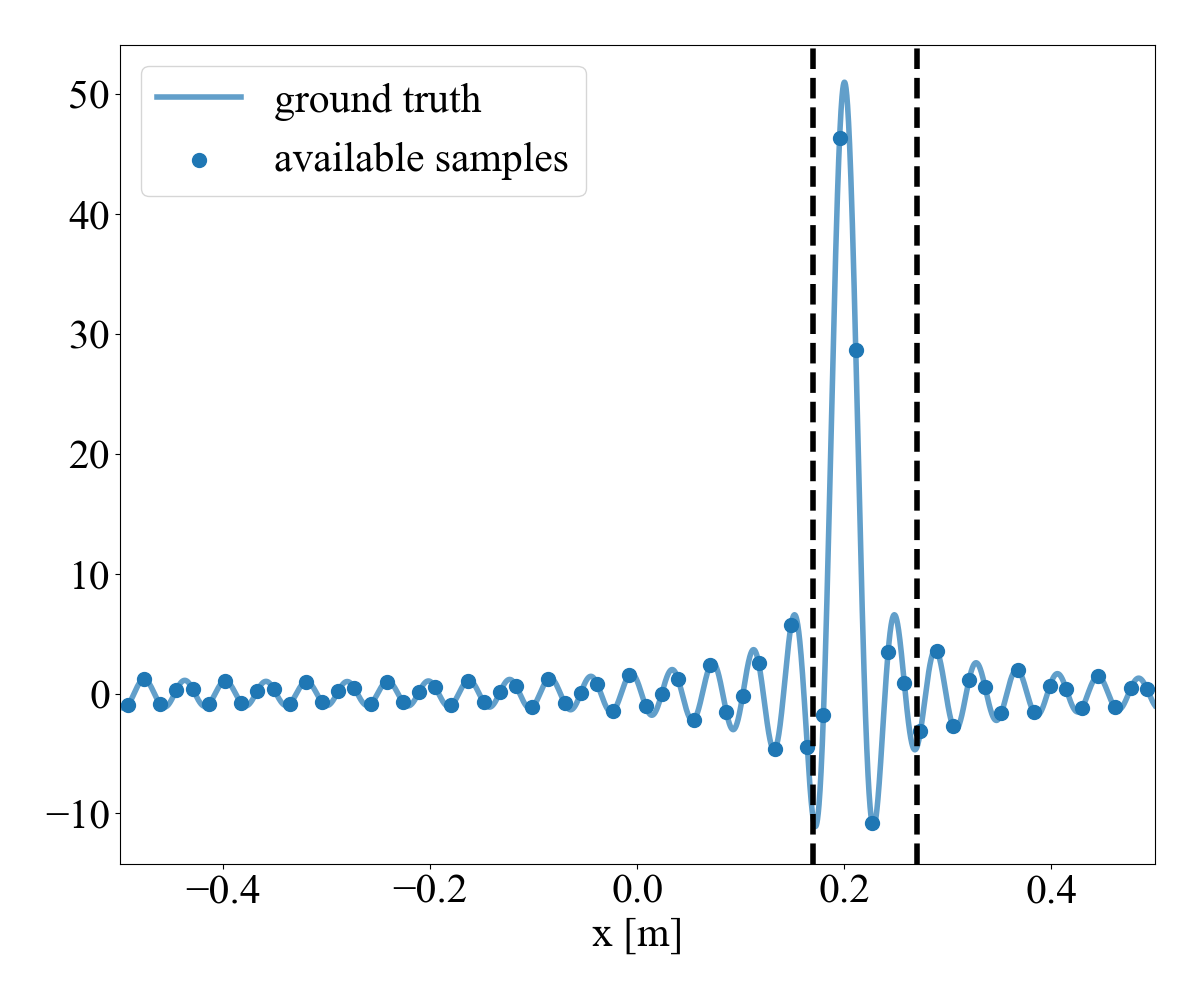} 
		\caption{1-D Dirichlet.}
		\label{fig:interp1d_input}
	\end{subfigure}
	\hfill
	\begin{subfigure}{.49\textwidth}
		\centering
		\includegraphics[width=0.99\linewidth]{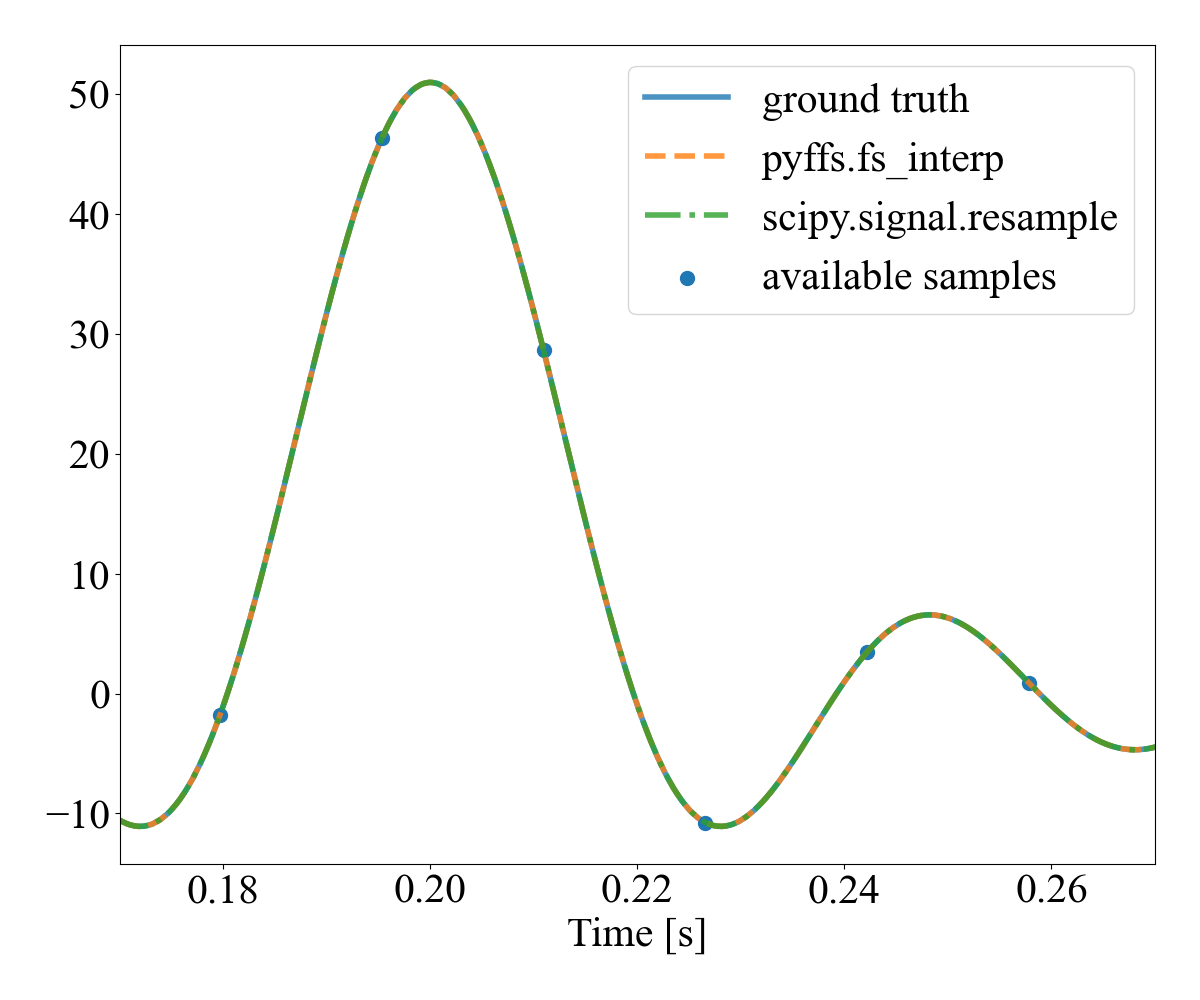}
		\caption{Interpolated section of (a).}
		\label{fig:interp1d_comparison}
	\end{subfigure}
	\centering
	\begin{subfigure}{.32\textwidth}
		\centering
		\includegraphics[width=0.99\linewidth]{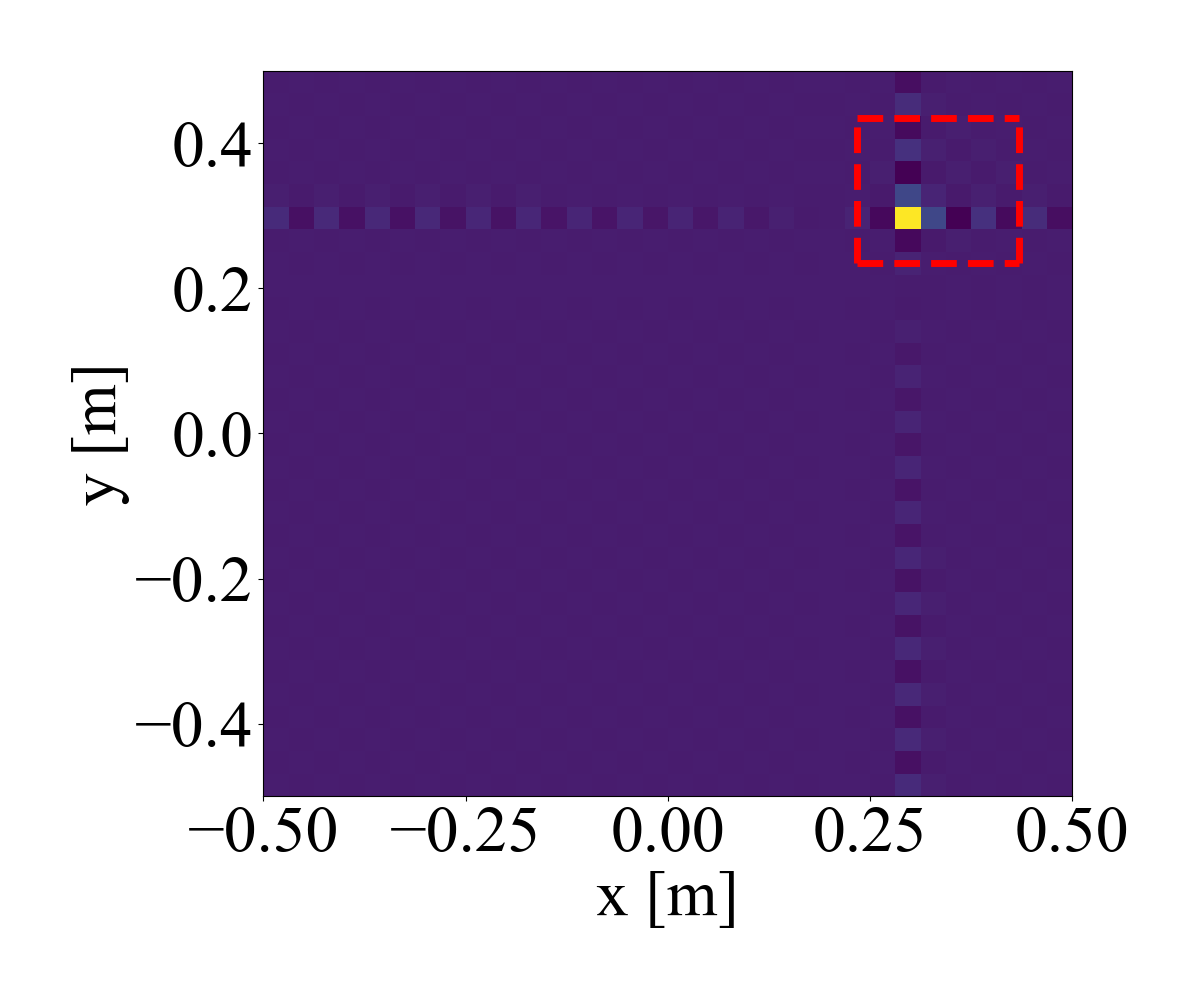} 
		\caption{2-D Dirichlet kernel.}
		\label{fig:interp_2d_input}
	\end{subfigure}
	\hfill
	\begin{subfigure}{.32\textwidth}
		\centering
		\includegraphics[width=0.99\linewidth]{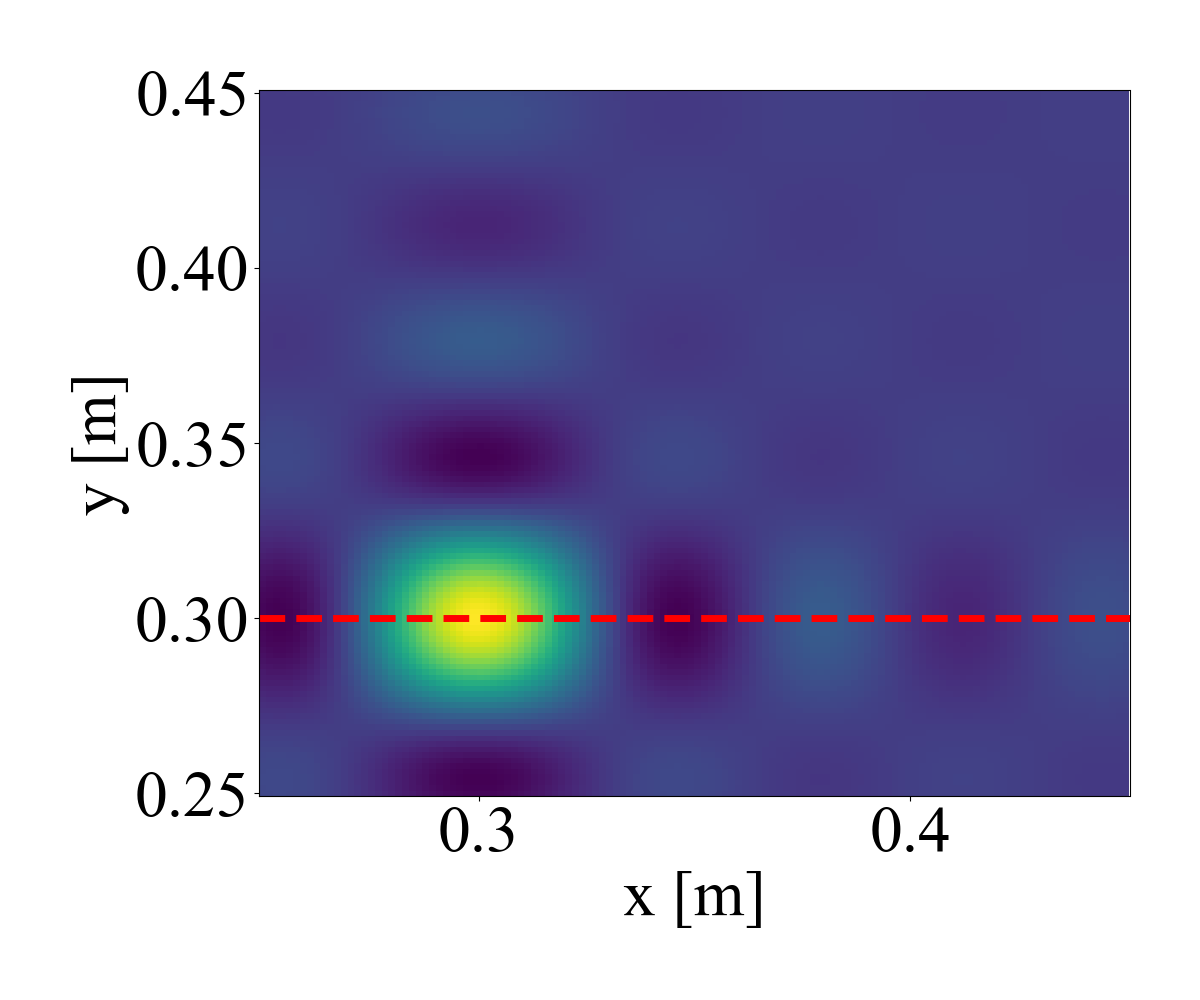}
		\caption{Interpolated region of (c).}
		\label{fig:interp_2d_ffs}
	\end{subfigure}
	\begin{subfigure}{.32\textwidth}
		\centering
		\includegraphics[width=0.99\linewidth]{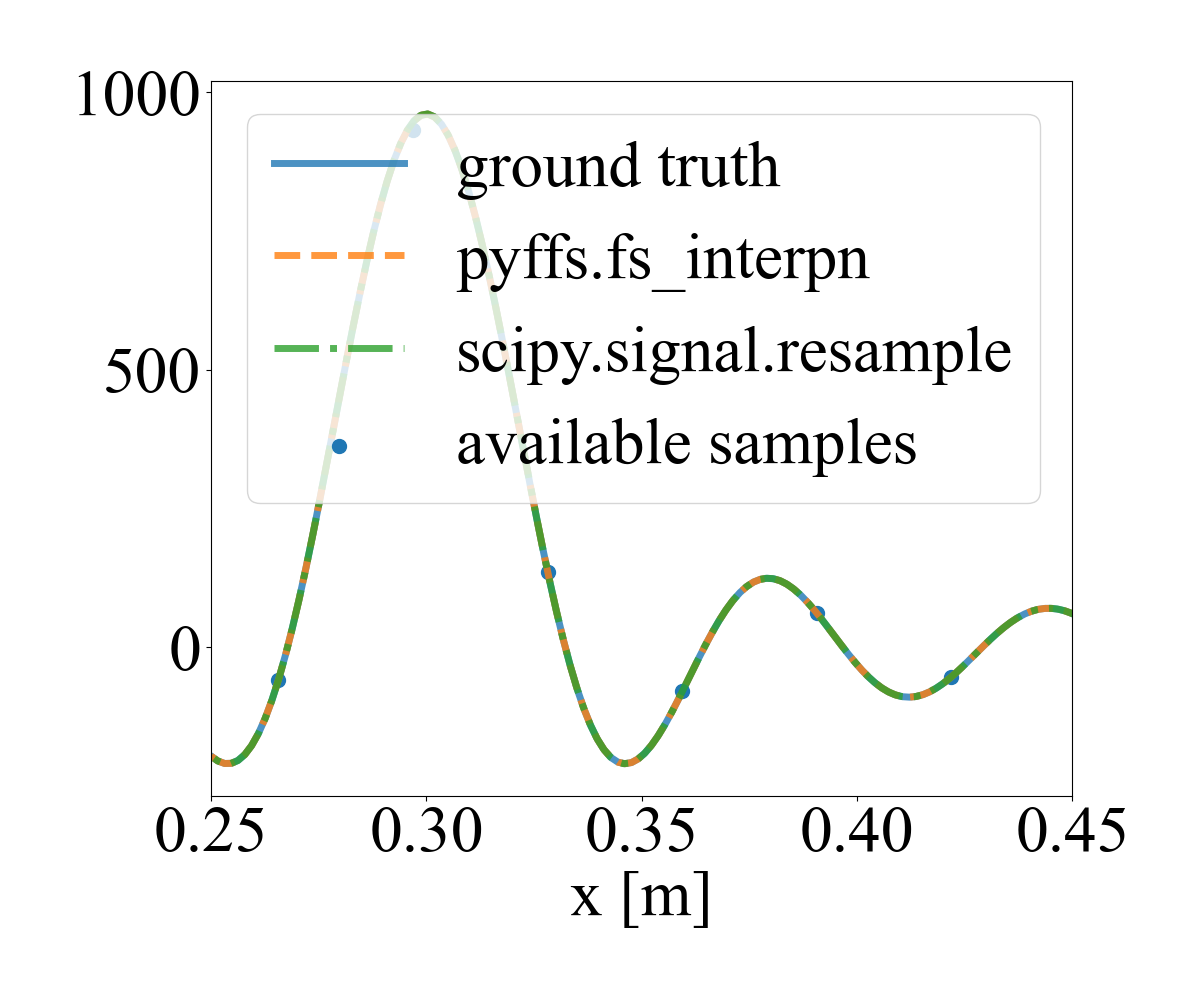}
		\caption{Cross-section at $ y=0.3 $.}
		\label{fig:interp_2d_output_slice}
	\end{subfigure}
	\caption{(Top) Bandlimited interpolation of a 1-D Dirichlet kernel, \Cref{eq:dirichlet} with $ N_{\text{FS}} = 51 $, $ T = 1 $, $ T_c  = 0$, and 64 samples. (Bottom) Bandlimited interpolation of a 2-D Dirichlet kernel, \Cref{eq:dirichlet_2d} with $ N_{\text{FS},x} = N_{\text{FS},y} = 31 $, $ T_x =T_y = 1 $, $ T_{c,x} = T_{c, y}  = 0$, and 256 samples along each dimension. The dashed box in (c) indicates the 2-D interpolated section shown in (d). The dashed line in (d) indicates the cross-section shown in (e).}
	\label{fig:interp_compare}
\end{figure}

\Cref{fig:interp_2d_input,fig:interp_2d_ffs,fig:interp_2d_output_slice} show the interpolation of a section of a 2-D Dirichlet function, which is essentially the outer product of two 1-D Dirichlet functions along the $ x$- and $ y$- dimensions
\begin{align}
	\label{eq:dirichlet_2d}
	\phi(x, y) &= \sum_{k_x = -N_x}^{N_x} \sum_{k_y = -N_y}^{N_y}
	\exp\left( j \frac{2 \pi}{T_x} k_x (x - T_{c,x}) \right)
	\exp\left( j \frac{2 \pi}{T_y} k_y (y - T_{c,y}) \right),
\end{align}
whose bandwidth is given by $ N_{\text{FS},x} = 2N_x + 1 $ and $ N_{\text{FS},y} = 2N_y + 1 $ in the $ x $- and $ y $- dimensions respectively. We can draw similar conclusions as for the 1-D case: the proposed technique and interpolating zero-padded DFT coefficients perfectly match the ground truth function as it is bandlimited.

\subsubsection*{Benchmark}

\begin{figure}[t!]
	\centering
	\begin{subfigure}{.49\textwidth}
		\centering
		\includegraphics[width=0.99\linewidth]{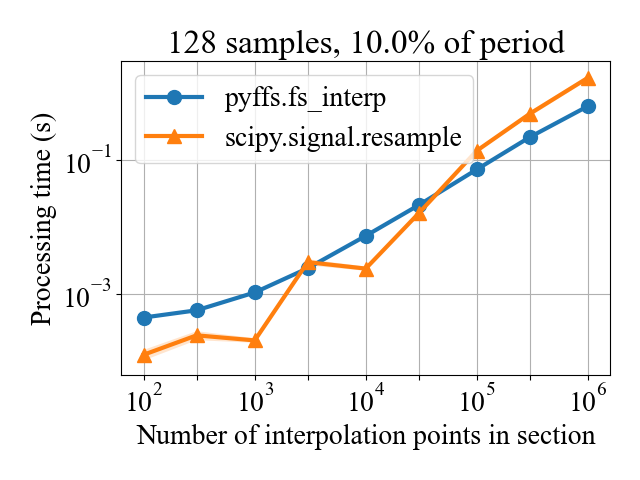}
		\caption{}
		\label{fig:bandlimited_interp1d_vary_M}
	\end{subfigure}
	\begin{subfigure}{.49\textwidth}
		\centering
		\includegraphics[width=0.99\linewidth]{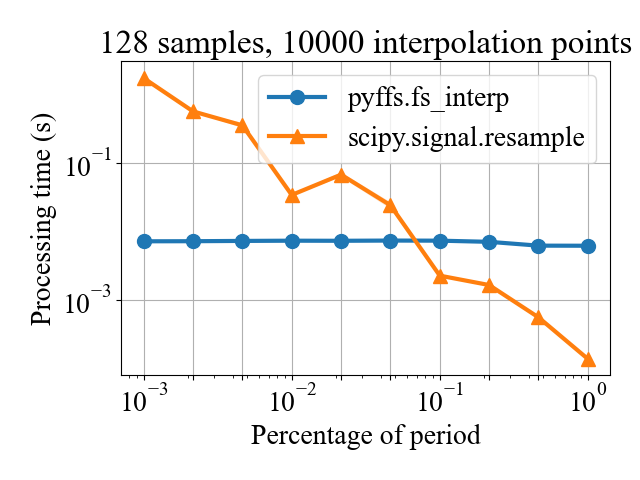}
		\caption{}
		\label{fig:bandlimited_interp1d_vary_width}
	\end{subfigure}
	\begin{subfigure}{.49\textwidth}
		\centering
		\includegraphics[width=0.99\linewidth]{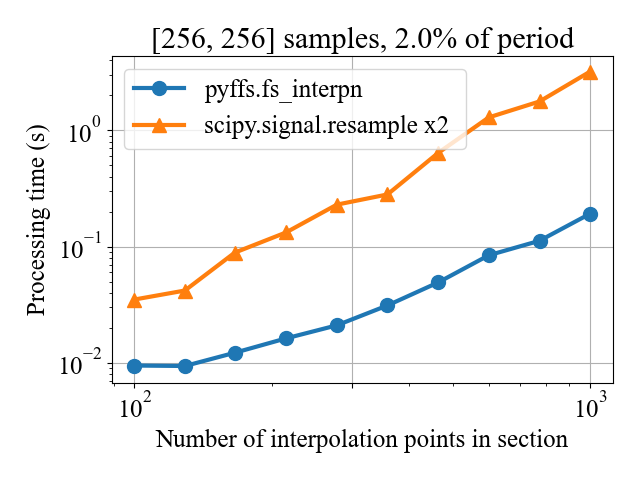}
		\caption{}
		\label{fig:bandlimited_interp2d_vary_M}
	\end{subfigure}
	\begin{subfigure}{.49\textwidth}
		\centering
		\includegraphics[width=0.99\linewidth]{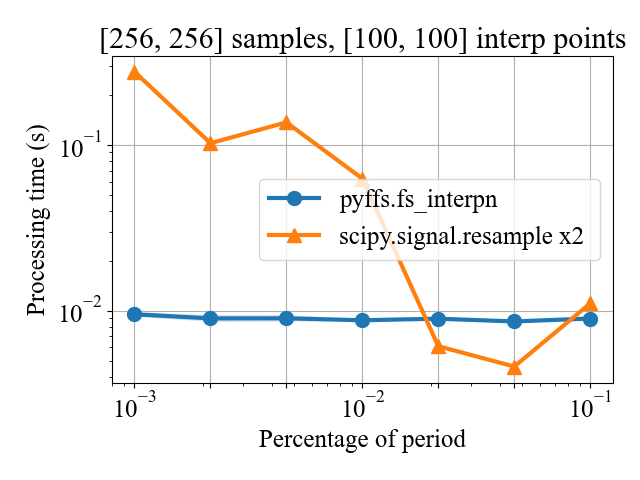}
		\caption{}
		\label{fig:bandlimited_interp2d_vary_width}
	\end{subfigure}
	\caption{(Top) Benchmarking 1-D bandlimited interpolation. $ 128 $ samples are taken of a signal that has bandwidth $ N_\text{FS} = 127 $. Each point in the curves is averaged over $ 10 $ trials. $ 10\% $ of the period is equivalent to the dotted black region in \Cref{{fig:interp1d_input}}. (Bottom) Benchmarking 2-D bandlimited interpolation. $ 256 \times 256 $ samples are taken of a signal that has bandwidth $ N_\text{FS} =[255 \times 255] $. Each point in the curves is averaged over $ 10 $ trials. $ 2\% $ of the period is equivalent to the dashed red box in \Cref{{fig:interp_2d_input}}.}
	\label{fig:benchmark_interp}
\end{figure}

%

As a reminder, the zero-padding technique results in a resampling across the entire period, whereas the proposed technique yields finer resolution only in the selected region through use of the CZT. As a result, the computational load of both approaches can be very different with respect to the width of the selected region and the number of interpolation points in that region.

In \Cref{fig:bandlimited_interp1d_vary_M,fig:bandlimited_interp1d_vary_width}, we compare the two bandlimited interpolation techniques in 1-D. As the number of interpolation points increases as shown in \Cref{fig:bandlimited_interp1d_vary_M}, it becomes slightly advantageous to use the CZT-based approach rather than interpolation by zero-padding DFT coefficients. 
The clear advantage of the CZT approach is when we focus in on smaller and smaller regions as shown in \Cref{fig:bandlimited_interp1d_vary_width},
as this necessitates a very large IDFT for the zero-padding approach. Moreover, for a fixed number of interpolation points, the computational load of the CZT approach is independent of the interpolation region size. As shown in \Cref{fig:bandlimited_interp1d_vary_width}, this is beneficial when we interpolate very small regions but detrimental when interpolating over larger sections. In the latter scenario, it is better to use the traditional zero-padding approach.

In \Cref{fig:bandlimited_interp2d_vary_M,fig:bandlimited_interp2d_vary_width}, we compare the two bandlimited interpolation techniques in 2-D. For a modest-sized input ($256 \times 256$ samples), it is noticeably advantageous to use the CZT-based approach as the number of interpolation points increases (\Cref{fig:bandlimited_interp2d_vary_M}), and very beneficial when we focus in on smaller regions (\Cref{fig:bandlimited_interp2d_vary_width}). The latter result makes the CZT-based approach an attractive choice when desiring to zoom into bandlimited images in a distortionless fashion.

\subsection{GPU acceleration}

\begin{figure}[t!]
	\centering
	\begin{subfigure}{.49\textwidth}
		\centering
		\includegraphics[width=0.99\linewidth]{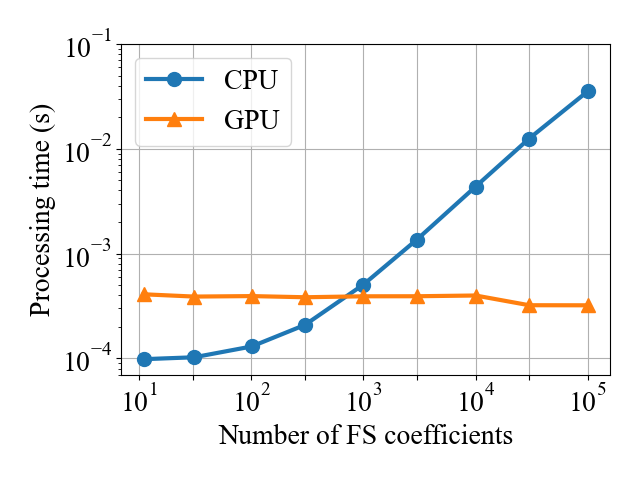} 
		\caption{1-D}
		\label{fig:ffs_1d}
	\end{subfigure}
	\hfill
	\begin{subfigure}{.49\textwidth}
		\centering
		\includegraphics[width=0.99\linewidth]{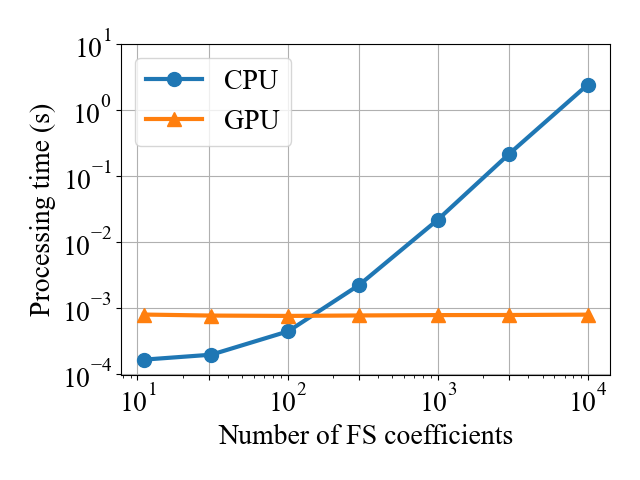}
		\caption{2-D}
		\label{fig:ffs_2d}
	\end{subfigure}
	\begin{subfigure}{.49\textwidth}
		\centering
		\includegraphics[width=0.99\linewidth]{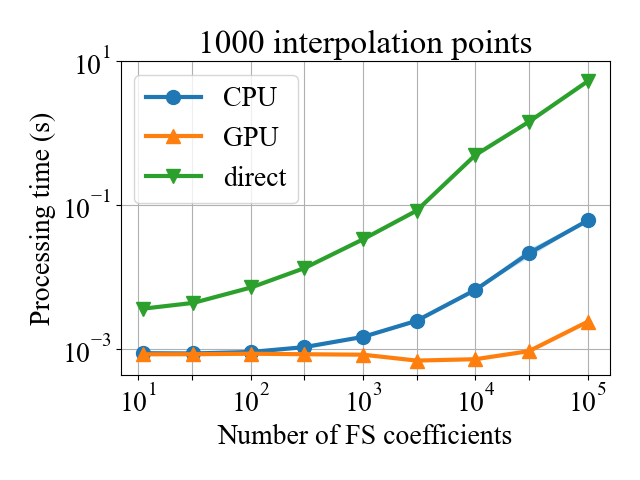} 
		\caption{}
		\label{fig:interp_1d_vary_Nfs}
	\end{subfigure}
	\hfill
	\begin{subfigure}{.49\textwidth}
		\centering
		\includegraphics[width=0.99\linewidth]{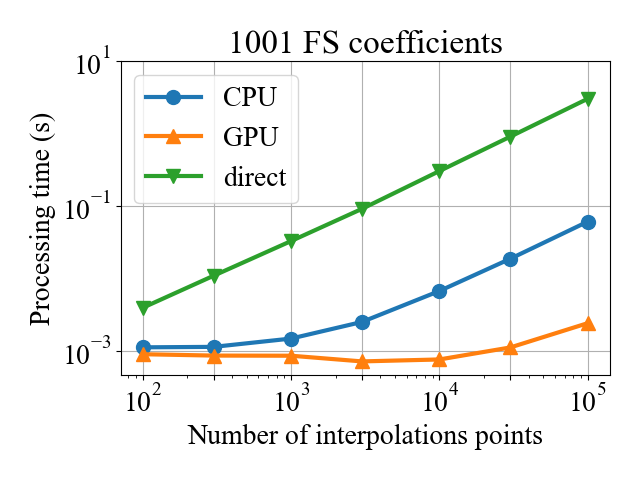}
		\caption{}
		\label{fig:interp1d_vary_M}
	\end{subfigure}
	\begin{subfigure}{.49\textwidth}
		\centering
		\includegraphics[width=0.99\linewidth]{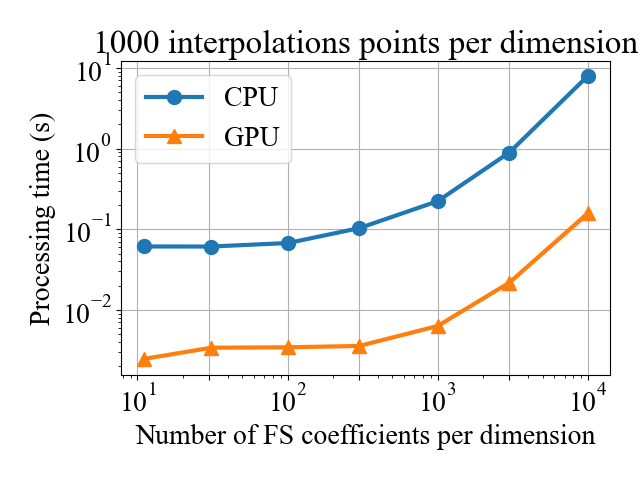} 
		\caption{}
		\label{fig:interp2d_vary_Nfs}
	\end{subfigure}
	\hfill
	\begin{subfigure}{.49\textwidth}
		\centering
		\includegraphics[width=0.99\linewidth]{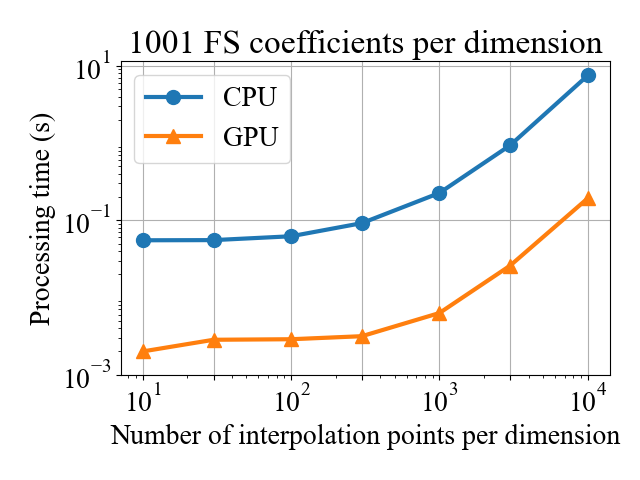}
		\caption{}
		\label{fig:interp2d_vary_M}
	\end{subfigure}
	\caption{Profiling GPU acceleration; each point is averaged over $ 10 $ trials. (Top) Fourier series computation; 1-D (middle) and 2-D (bottom) Fourier series interpolation.}
\end{figure}

We now quantify the speed-up provided by a GPU for FS computation and interpolation. There are two important considerations when using a GPU. Firstly, if the application permits, it is recommended to work with \texttt{float32} / \texttt{complex64} arrays for less memory consumption and potentially faster computation. By default, NumPy and CuPy create \texttt{float64} / \texttt{complex128} arrays, e.g. when initializing an array with \texttt{np.zeros}, so casting the arrays accordingly is recommended. In the benchmarking tests below, we use \texttt{float32} / \texttt{complex64} arrays. Secondly, the benefits of using a GPU typically emerge when the processed arrays are larger than the CPU cache. So the crossover between CPU and GPU performance can be very hardware dependent. 

\Cref{fig:ffs_1d,fig:ffs_2d} compares the processing time between a CPU and a GPU for computing an increasing number of FS coefficients. 
In 1-D, for more than $ 1'000 $ coefficients it starts to become beneficial to use a GPU, and at around $ 10'000 $ coefficients it is an order of magnitude faster to use a GPU.
In 2-D, the crossover point is at around $ 100 $ coefficients per dimension, and at around $ 1'000 $ coefficients per dimension it is more than an order of magnitude faster to use a GPU. From the 1-D and 2-D cases, it is clear that using a GPU scales well as the input increases in size. When considering a 2-D or even a 3-D object, where input sizes quickly grow, it is attractive to make use of a GPU for even modest input sizes.

\Cref{fig:interp_1d_vary_Nfs,fig:interp1d_vary_M} profiles the processing time for 1-D FS interpolation. The following three approaches are compared:
\begin{itemize}
	\item Directly evaluating the bandlimited Fourier synthesis expression~\cref{eq:bl_ifs} at each timestamp on a CPU.
	\item Applying \Cref{thm:fast_interp_complex} on a CPU.
	\item Applying \Cref{thm:fast_interp_complex} on a GPU.
\end{itemize}
As we vary both the number of FS coefficients and the number of interpolation points, using \Cref{thm:fast_interp_complex} greatly reduces computational cost, as we observe an order of magnitude reduction for the
interpolation of $ 300 $ FS coefficients.
The difference between ``direct'' and ``CPU'' is essentially the gains we get from the FFT algorithm, namely $ \mathcal{O}(N \log N) $ complexity instead of $ \mathcal{O}(N^3) $.
Using a GPU becomes more attractive as the number of coefficients and number of samples exceeds $ 300 $. As mentioned earlier, this is probably the point when the arrays and computation can no longer fit on the CPU cache.

\Cref{fig:interp2d_vary_Nfs,fig:interp2d_vary_M} compares the CPU and GPU approaches in 2-D function. Using a GPU consistently provides two orders of magnitude faster computation for a varying number of FS coefficients and varying number of interpolation points per dimension. The direct approach is not even considered as it is much too slow. The benefits of using a GPU are even more prominent in 2-D as input sizes quickly grow when considering multidimensional scenarios.

At the time of writing, CuPy has not implemented an equivalent of SciPy's \texttt{resample} function to perform interpolation comparisons as in \Cref{fig:benchmark_interp}.


%% file: content/5_fourier_optics.tex
\section{Example application in Fourier optics} 
\label{sec:results}


\begin{figure}[t!]
	\centering
	\begin{subfigure}{.34\textwidth}
		\centering
		\includegraphics[width=0.99\linewidth]{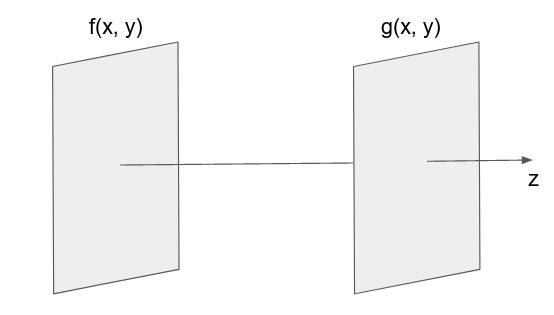} 
		\caption{}
		\label{fig:wave_prop}
	\end{subfigure}
	\hfill
	\begin{subfigure}{.34\textwidth}
		\centering
		\includegraphics[width=0.95\linewidth]{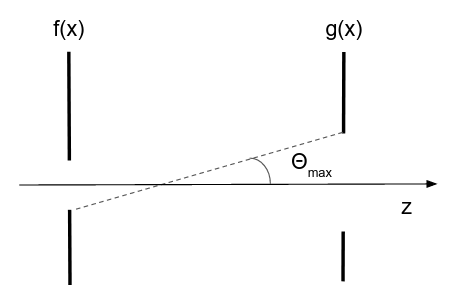}
		\caption{}
		\label{fig:max_angle}
	\end{subfigure}
	\hfill
	\begin{subfigure}{.3\textwidth}
		\centering
		\includegraphics[width=0.87\linewidth]{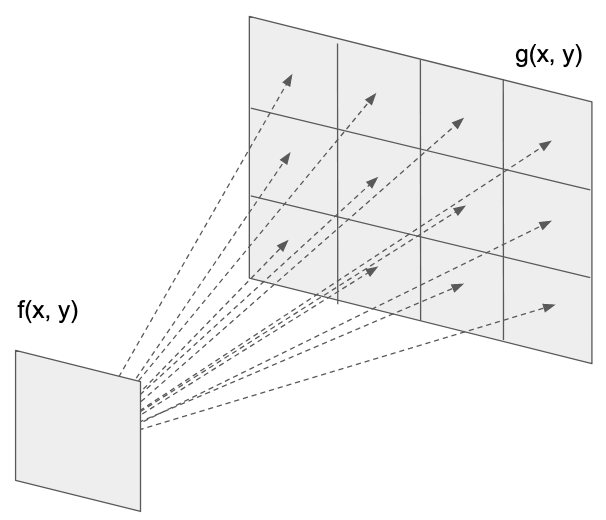}
		\caption{}
		\label{fig:rect_tiling}
	\end{subfigure}
	\caption{Visualization of optical wave propagation setup. (\ref{fig:wave_prop}) When simulating optical wave propagation for holography, one often considers the propagation along the $ z $-axis between two parallel planes, one being the source - $ f(x, y) $ - and the other being the target plane - $ g(x, y) $. (\ref{fig:max_angle}) In practice we have a finite input and output region, which determines the maximum angle and therefore maximum spatial frequency we can observe. Focusing on a single axis $ x $, this maximum frequency is given by $ \sin \theta_{\text{max}} / \lambda $, where $ \lambda $ is the optical wavelength. (\ref{fig:rect_tiling}) Lower resolution source wavefront for holographic tiling.}
	\label{fig:optics}
\end{figure}

In Fourier optics, we are often interested in the propagation of light between two planes, i.e.\ a source plane and a target plane as shown in \Cref{fig:wave_prop}. Given an aperture function or phase pattern at the source plane, we would like to determine the pattern at the target plane, as predicted by the Rayleigh-Sommerfeld diffraction formula. This propagation is often modeled with one of three approaches that make use of the FFT for an efficient simulation: Fraunhofer approximation, Fresnel approximation, or the angular spectrum method~\cite{Goodman2005}. The choice between these three approaches typically depends on the requirements of the application, e.g.\ the distance between the two planes
and the size of input and output regions~\cite{Schmidt2010}. 
For all approaches, we again find ourselves with a continuous-domain phenomenon that can be considered bandlimited and periodic. Bandlimited as in practice we consider finite input and output regions, lending to a restricted set of angles and therefore a bandlimited spatial frequency response between the source and target planes. This restriction of angles is shown in \Cref{fig:max_angle}. Even though our input may not be bandlimited, the resulting output is bandlimited after convolution with such a response~\cite{Matsushima2009}. Finally, we can frame the optical simulation as periodic as the input and output regions have a compact support and can thus be replicated to form periodic signals. 

The application of the CZT, or equivalently the fractional FT, for interpolation has already found its use in Fourier optics to resample the output plane outside of the grid defined by the FFT~\cite{Muffoletto2007,Nascov2009,Yu2012}, as demonstrated with pyFFS in \cref{fig:interp_2d_input,fig:interp_2d_ffs}. 

Below we show how the pyFFS interface can be used in optical wave propagation for efficient simulation and interpolation.

\begin{lstlisting}[language=Python, label={lst:prop}, caption=Optical free space propagation with pyFFS.]
# pad input and reorder
f_pad = numpy.pad(f, pad_width=pad_width)
f_pad_reorder = pyffs.ffs_shift(f_pad)

# compute FS coefficients of input
F = pyffs.ffsn(f_pad_reorder, T, T_c, N_FS)

# convolution in frequency domain with free space transfer function
G = F * H

# interpolate at the desired location and resolution
# a and b specify the region while N_out specifies the resolution
g = pyffs.fs_interpn(G, T, a, b, N_out)
\end{lstlisting}

The free space propagation transfer function \texttt{H} in the above code listing can be obtained by evaluating the analytic expression for the Fresnel approximation or the angular spectrum method transfer functions at the appropriate frequency values~\cite{Goodman2005}, or by measuring this response and computing its FS coefficients with \texttt{pyffs.ffsn}.

One may wish to simulate an output window with the same size as the input but at a finer resolution.
In order to circumvent the much larger FFT that this may require,
an approach known as rectangular tiling~\cite{Muffoletto2007}, as shown in \Cref{fig:rect_tiling}, can be used to split the output window into tiles. In its original proposition, the tiles were simulated sequentially, but with a GPU they could be computed in parallel for a significantly shorter simulation time: pyFFS's GPU support enables this possibility. Moreover, rectangular tiling in its original proposition requires that each tile has the same number of samples as the input window. This restriction is removed by the interpolation approach of pyFFS.


%% file: content/6_conclusion.tex
\section{Conclusion}
\label{sec:conclusion}

In this paper we have presented pyFFS, a Python library for efficient Fourier series (FS) coefficient computation, convolution, and interpolation. The intended use of this package is when working with discrete samples that arise from a continuous-domain signal. When the underlying signal is periodic (or has finite support and can be periodized) and bandlimited, its FS coefficients can be computed and interpolated in a straightforward and distortionless fashion with pyFFS. If either periodicity or bandlimitedness is not met, the same workarounds as when applying the discrete Fourier transform can be used, namely windowing to taper discontinuous boundaries or bandlimiting by FS coefficient truncation.

As computation is posed in the continuous-domain, accuracy loss that may arise from switching between the discrete- and the continuous-domain can be minimized. Moreover, this package serves as a handy continuous-domain complement to the functionalities already available in SciPy~\cite{Virtanen2020}.
We also provide functionality not available in SciPy, namely N-D circular convolution,  N-D bandlimited interpolation, and a bandlimited interpolation technique based on the chirp Z-transform. As shown in our benchmarking results, the latter can be more than an order of magnitude faster when interpolating sub-regions of a 1-D  or 2-D periodic function. Similar results can be expected for a general N-D function. Furthermore, GPU support has been seamlessly integrated through the CuPy package~\cite{Okuta2017}, offering more than an order of magnitude reduction when computing and interpolating a large number of FS coefficients.

In summary, pyFFS offers researchers and engineers a convenient and efficient interface for working with FS coefficients. The source code is made available on GitHub\footnote{\rurl{github.com/imagingofthings/pyFFS}} and can be easily installed for Python through PyPi.\footnote{\texttt{pip install pyffs}} More extensive and up-to-date documentation can be found at \rurl{pyffs.readthedocs.io}.

%% file: content/appendix.tex
\appendix
 \section{Fast Fourier series computation for even-length sequences and proofs} 
\label{sec:proofs}

 \Cref{thm:ffs_odd} addresses the fast Fourier series computation for odd-length sequences. For even-length sequences, there is a slight modification in the timestamps and modulation terms.

\begin{theorem}[Fast Fourier series, $N_{\text{s}} \in 2 \bN$]\label{thm:ffs_even}
      	Let $x : \bR \to \bC$ be a $T$-periodic function of bandwidth $N_{\text{FS}} = 2 N +
      	1$, with $T_{c} \in \bR$ the mid-point of any period.  Let $Q \in 2 \bN + 1$ be an
      	arbitrary odd integer such that $N_{\text{s}} = N_{\text{FS}} + Q$. Then
	\begin{gather}
		\bbx = N_{\text{s}} \iDFT_{N_s}\bigParen{\bbX^{\text{FS}} \odot B_{1}^{\bbE_{1}}} \odot B_{2}^{N \bbE_{2}}, \label{eq:ffs_even_phi} \\
		\bbX^{\text{FS}} = \frac{1}{N_{\text{s}}} \DFT_{N_s}\bigParen{\bbx \odot B_{2}^{-N \bbE_{2}} } \odot B_{1}^{-\bbE_{1}}, 
		\label{eq:ffs_even_phiFS}
	\end{gather}
      	where
      	\begin{gather}
      		\bbx = \bigBrack{x(t_{0}), \ldots, x(t_{M-1}), x(t_{-M}), \ldots, x(t_{-1})} \in \bC^{N_{\text{s}}}, \nonumber \\
      		\bbX^{\text{FS}} = \bigBrack{X_{-N}^{\text{FS}}, \ldots, X_{N}^{\text{FS}}, \bbZero_{Q}} \in \bC^{N_{\text{s}}}, \nonumber \\
      		t_{n} = T_{c} + \frac{T}{N_{\text{s}}} \bigParen{\frac{1}{2} + n}, \quad n \in \bZ, \label{eq:ffs_even_tn} \\
      		M = N_{\text{s}} / 2, \nonumber
      	\end{gather}
      	and
      	
      		\resizebox{0.95\columnwidth}{!}{
      	\begin{minipage}{0.45\textwidth}
      		\begin{center}
      			\begin{gather*}
      				B_{1} = \exp\bigParen{j \frac{2 \pi}{T} \bigBrack{T_{c} + \frac{T}{2 N_{\text{s}}}}} \in \bC, \\
      				\bbE_{1} = \bigBrack{-N, \ldots, N, \bbZero_{Q}} \in \bZ^{N_{\text{s}}},
      			\end{gather*}
      		\end{center}
      	\end{minipage}
      	\begin{minipage}{0.5\textwidth}
      		\begin{center}
      			\begin{gather*}
      				B_{2} = \exp\bigParen{-j \frac{2 \pi}{N_{\text{s}}}} \in \bC, \\
      				\bbE_{2} = \bigBrack{0, \ldots, M-1, -M, \ldots, -1} \in \bZ^{N_{\text{s}}}.
      			\end{gather*}
      		\end{center}
      	\end{minipage}
      }
\end{theorem}
 
\begin{proof}{\cref{thm:ffs_odd,thm:ffs_even}}
	
	Starting with the Fourier series (FS) synthesis expression for a bandlimited signal~\cref{eq:bl_ifs}, we plug in $t_{n}$ from \cref{eq:ffs_odd_tn,eq:ffs_even_tn}
	\begin{align*}
		x(t_{n})
		& = \sum_{k = -N}^{N} X_{k}^{\text{FS}} \exp\bigParen{j \frac{2 \pi}{T} k t_{n}} \\
		& = \begin{dcases}
			\sum_{k = -N}^{N} X_{k}^{\text{FS}} \exp\bigParen{j \frac{2 \pi}{T} k  \bigBrack{T_{c} + \frac{T}{N_{\text{s}}} n } }, \quad \text{\cref{thm:ffs_odd}} \\
			\sum_{k = -N}^{N} X_{k}^{\text{FS}} \exp\bigParen{j \frac{2 \pi}{T} k  \bigBrack{T_{c} +  \frac{T}{N_{\text{s}}} \bigParen{\frac{1}{2} + n} } }, \quad \text{\cref{thm:ffs_even}}.
		\end{dcases}
	\end{align*}
	With $ B_{1} = \exp\bigParen{j \frac{2 \pi}{T} T_{c}} $ for \cref{thm:ffs_odd} and $ B_{1} = \exp\bigParen{j \frac{2 \pi}{T} \bigBrack{T_{c} + \frac{T}{2 N_{\text{s}}}}} $ for \cref{thm:ffs_even}, and $ B_{2} = \exp\bigParen{-j \frac{2 \pi}{N_{\text{s}}}} $ for both cases, we can write
	\begin{align*}
		x(t_{n}) &= \sum_{k = -N}^{N} X_{k}^{\text{FS}} B_{1}^{k} B_{2}^{-nk}.
	\end{align*}
	We can then shift the summation terms so they are similar to that of an inverse discrete Fourier transform (IDFT), i.e.\ summation starting at $ k = 0 $
	\begin{align*}
		x(t_{n}) &= \sum_{k = 0}^{2 N} X_{k - N}^{\text{FS}} B_{1}^{k - N} B_{2}^{-n(k-N)} = B_{2}^{n N} \sum_{k = 0}^{N_{\text{s}} - 1} X_{k - N}^{\text{FS}} B_{1}^{k - N} B_{2}^{-nk},
	\end{align*}
	where $ N_{\text{s}} = 2N + 1 $.
	
	We now introduce notation to index the $ k $-th element of a vector as $ x_k = [\bbx]_k $ in order to write
	\begin{align*}
		x(t_{n}) & = B_{2}^{n N} \sum_{k = 0}^{N_{\text{s}} - 1} \bigBrack{\bbX^{\text{FS}} \odot B_{1}^{\bbE_{1}}}_{k} B_{2}^{-nk},
	\end{align*} 
	where $\bbE_{1}$ is as defined in \cref{thm:ffs_odd,thm:ffs_even}, and $ \bbX^{\text{FS}}  $ contains the ordered FS coefficients of $ x(t) $ within $ [-N, N] $. 
	
	As $ B_{2} = \exp\bigParen{-j \frac{2 \pi}{N_{\text{s}}}} $, we can write the above expression as an IDFT of modulated FS coefficients
	\begin{align*}
		x(t_{n}) = N_{\text{s}} \bigBrack{\iDFT_{N_{\text{s}}}(\bbX^{\text{FS}} \odot B_{1}^{\bbE_{1}})}_{n} B_{2}^{n N}.
	\end{align*} 
	Using the periodicity of $ x(t) $, namely $x(t_{n}) = x(t_{n + q N_{s}})$ for $ q \in \mathbb{Z} $, we can write\\[0.2em]
    	\begin{center}
    		\small
    		\fbox{
    			\begin{minipage}{0.6\textwidth}
    				\begin{center}
    					\vspace{0.25em}
    					\Cref{thm:ffs_odd}
    					\begin{align*}
    						\bbx
    						& = \bigBrack{x(t_{0}), \ldots, x(t_{M}), x(t_{M + 1}), \ldots, x(t_{2 M})} \\
    						& = \bigBrack{x(t_{0}), \ldots, x(t_{M}), x(t_{-M}), \ldots, x(t_{-1})}.
    					\end{align*}
    				\end{center}
    			\end{minipage}
    		}
    		\fbox{
    			\begin{minipage}{0.6\textwidth}
    				\begin{center}
    					\vspace{0.25em}
    					\Cref{thm:ffs_even}
    					\begin{align*}
    						\bbx
    						& = \bigBrack{x(t_{0}), \ldots, x(t_{M - 1}), x(t_{M}), \ldots, x(t_{2 M - 1})} \\
    						& = \bigBrack{x(t_{0}), \ldots, x(t_{M - 1}), x(t_{-M}), \ldots, x(t_{-1})}.
    					\end{align*}
    				\end{center}
    			\end{minipage}
    		}
    	\end{center}
    	\vspace{1em}
    	We can write this set of reordered samples as the IDFT of modulated FS coefficients
    	\begin{align*}
    		\bbx = N_{\text{s}} \iDFT_{N_{\text{s}}}(\bbX^{\text{FS}} \odot B_{1}^{\bbE_{1}}) B_{2}^{N \bbE_{2}},
    	\end{align*}
    	where $ \bbE_{2} ,$ as defined in \cref{thm:ffs_odd,thm:ffs_even}, encapsulates this reordering of samples. We have therefore shown the relations \cref{eq:ffs_odd_phi,eq:ffs_even_phi} for \cref{thm:ffs_odd,thm:ffs_even} respectively. In order to obtain \cref{eq:ffs_odd_phiFS,eq:ffs_even_phiFS}, which express the ordered FS coefficients as a function of the reordered samples, we simply have to invert \cref{eq:ffs_odd_phi,eq:ffs_even_phi} respectively.
\end{proof}
 
 \section{Fast Fourier series interpolation proof} 
 \label{sec:fs_interp_proof}
 
\begin{proof}{\cref{thm:fast_interp_complex}}
	
	Starting with the Fourier series (FS) synthesis expression for a bandlimited signal~\cref{eq:bl_ifs}, we plug in $t_{n}$ from \cref{thm:fast_interp_complex}
	
	\begin{align*}
	x(t_{n})
	& = \sum_{k = -N}^{N} X_{k}^{\text{FS}} \exp\bigParen{j \frac{2 \pi}{T} k t_{n}} = \sum_{k = -N}^{N} X_{k}^{\text{FS}} \exp\bigParen{j \frac{2 \pi}{T} k  \Big( a + \frac{b - a}{M - 1} n \Big)  }.
	\end{align*}
	Using $ A = \exp\bigParen{-j \frac{2\pi}{T} a}$ and $W = \exp\bigParen{j \frac{2 \pi}{T} \frac{b - a}{M - 1}} $ from \cref{thm:fast_interp_complex}, we can write
	\begin{align*}
	x(t_{n})
	& = \sum_{k = -N}^{N} X_{k}^{\text{FS}} A^{-k} W^{n k}.
	\end{align*}
	We can then shift the summation terms so they are similar to that of the chirp Z-transform (CZT) from \cref{def:czt}, i.e.\ summation starting at $ k = 0 $
	\begin{align*}
	x(t_{n}) & =  \sum_{k = 0}^{2 N} X_{k - N}^{\text{FS}} A^{-(k-N)} W^{n(k-N)} =  A^{N} W^{-Nn} \sum_{k = 0}^{N_{\text{FS}} - 1} X_{k - N}^{\text{FS}} A^{-k} W^{nk},
	\end{align*}
	where $ N_{\text{FS}} = 2N + 1 $.
	
	We can then use \cref{def:czt} of the CZT to writex
	\begin{align*}
	x(t_{n})
	= A^{N} W^{-Nn} \bigBrack{\CZT_{N_{\text{FS}}}^{M}(\bbX^{FS})}_{n}.
	\end{align*}
	Rearranging $x(t_{n})$ into vector form concludes the proof.
\end{proof}